\documentclass[abstract=on]{scrartcl}

\usepackage[bindingoffset=0mm,left=25mm,right=25mm,top=25mm,bottom=25mm,footskip=14mm]{geometry}
\usepackage[dvipsnames]{xcolor}
\usepackage{titling}
\usepackage{titlesec}
\usepackage{amsmath}
\usepackage{amssymb}
\usepackage[utf8]{inputenc}
\usepackage{url}
\usepackage{authblk}
\usepackage{enumitem}
\usepackage{graphicx}
\usepackage{listings}

\usepackage{caption}
\makeatletter
\renewcommand{\@seccntformat}[1]{%
  \ifcsname format@#1\endcsname
    \csname format@#1\endcsname
  \else
    \csname the#1\endcsname\quad %
  \fi
}
\g@addto@macro\appendix{%
  \def\format@section{Appendix \thesection: }%
}
\makeatother

\usepackage[nameinlink]{cleveref}
\crefname{lstlisting}{listing}{listings}

\usepackage{tikz}
\usetikzlibrary{calc, patterns}
\usepackage[normalem]{ulem}
\newcommand{\soutthick}[1]{%
    \renewcommand{\ULthickness}{1.2pt}%
       \sout{#1}%
    \renewcommand{\ULthickness}{.4pt}%
}
\usepackage{pgfplots}

\definecolor{HighlightColor}{rgb}{0.43, 0.69, 0.26}
\definecolor{HighlightColor2}{rgb}{0.945, 0.349, 0.373}
\definecolor{HighlightColor3}{rgb}{0.349, 0.604, 0.827}
\definecolor{HighlightColor4}{rgb}{0.976, 0.651, 0.353}
\definecolor{HighlightColor5}{rgb}{.62, 0.4, 0.671}
\definecolor{HighlightColor6}{rgb}{.804, 0.439, 0.345}
\definecolor{HighlightColor7}{rgb}{.843, 0.498, 0.702}
\definecolor{HighlightColor8}{rgb}{0.745, 0.769, 0.349}
\definecolor{CodeBG}{rgb}{0.95,0.95,0.95}

\definecolor{NVGreen}{rgb}{0.43, 0.69, 0.26} %
\definecolor{NVRed}{rgb}{0.945, 0.349, 0.373}
\definecolor{NVBlue}{rgb}{0.349, 0.604, 0.827}
\definecolor{NVOrange}{rgb}{0.976, 0.651, 0.353}
\definecolor{NVPurple}{rgb}{.62, 0.4, 0.671}
\definecolor{NVBrown}{rgb}{.804, 0.439, 0.345}
\definecolor{NVPink}{rgb}{.843, 0.498, 0.702}
\definecolor{NVLime}{rgb}{0.745, 0.769, 0.349}

\usepackage{listings}
\usepackage{algorithm2e}

\def\LRA{\Leftrightarrow\mkern40mu}
\def\LA{\Leftarrow\mkern40mu}

\usepackage{amsthm}

\theoremstyle{definition} %
\newtheorem{lemma}{Lemma}
\renewenvironment{proof}{{\bfseries Proof.}}{$\square$}

\usepackage{graphicx}
\graphicspath{{../Images/}}

\begin{document}

\title{\vspace*{-4em}Fast, High Precision Ray/Fiber Intersection \protect\\using Tight, Disjoint Bounding Volumes}
\date{}
\author{Nikolaus Binder, Alexander Keller\\NVIDIA}
\maketitle

\begin{figure}[!ht]
	\centering
	\vspace*{-4em}
	\begin{tabular}[t]{@{}ccc@{}}
	\begin{tikzpicture}
		\begin{axis}[
			height = .25\linewidth,
			width = .36\linewidth,
			xlabel = subdivision depth,
			ylabel = G rays/s,
			y label style={yshift=-3ex},
			ymin  = 0,
			ymax = 14,
			xmin = 1,
			xmax = 34,
			xtick ={2, 4, 9, 16, 22},
			xtick pos=left,
			ytick pos=left,
			scaled y ticks = false,
			legend style={font=\small,draw=none,fill=none}
		]
			\addplot[color=NVRed, mark=*, mark size=1, thick] coordinates {
			(2, 10.500)
			(3, 9.100)
			(4, 7.820)
			(5, 6.300)
			(6, 4.800)
			(7, 3.400)
			(8, 2.200)
			(9, 1.300)
			(10, 0.740)
			(11, 0.414)
			(12, 0.206)
			(13, 0.105)
			(14, 0.053)
			(15, 0.026) %
			(16, 0.013)
			(17, 0.006)
			(18, 0.003)
			(19, 0.001)
			(20, 0.0)
			(21, 0.0)
			(22, 0.0)
			};
			\addplot[color=NVGreen, mark=square*, mark size=1, thick] coordinates {
				(2, 8.880)
				(3, 8.380)
				(4, 7.880)
				(5, 7.420)
				(6, 7.000)
				(7, 6.650)
				(8, 6.300)
				(9, 6.000)
				(10, 5.760)
				(11, 5.520)
				(12, 5.300)
				(13, 5.100)
				(14, 4.920)
				(15, 4.720)
				(16, 4.470)
				(17, 4.200)
				(18, 4.020)
				(19, 3.840)
				(20, 3.740)
				(21, 3.680)
				(22, 3.610)
			};
			\node[anchor=south west] at (axis cs: 23.5, 0) {\includegraphics[width=.05\textwidth, trim = 140px 0 200px 0]{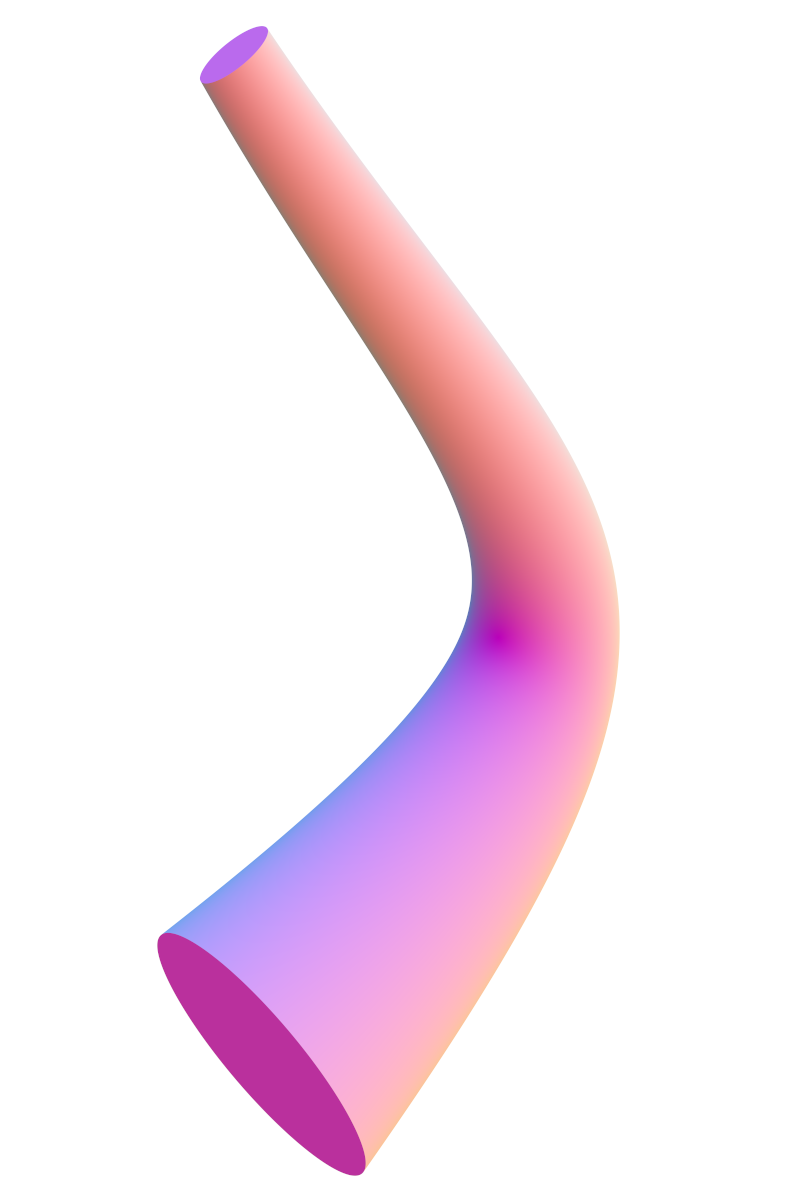}};
			\draw[densely dashed] (axis cs: 4,0) -- (axis cs: 4, 7.880);
		\end{axis}
	\end{tikzpicture}%
	&
	\begin{tikzpicture}
		\begin{axis}[
			height = .25\linewidth,
			width = .37\linewidth,
			xlabel = subdivision depth,
			ymin  = 0,
			ymax = 14,
			xmin = 1,
			xmax = 32,
			xtick ={2, 9, 16, 22},
			xtick pos=left,
			ytick pos=left,
			scaled y ticks = false,
			legend style={font=\small,draw=none,fill=none}
		]
			\addplot[color=NVRed, mark=*, mark size=1, thick] coordinates {
			(2, 13.100)
			(3, 12.600)
			(4, 12.300)
			(5, 11.800)
			(6, 11.100)
			(7, 10.300)
			(8, 9.300)
			(9, 8.160)
			(10, 6.560)
			(11, 4.330)
			(12, 2.600)
			(13, 1.440)
			(14, 0.765)
			(15, 0.395)
			(16, 0.200)
			(17, 0.102)
			(18, 0.051) %
			(19, 0.025)
			(20, 0.012)
			(21, 0.006)
			(22, 0.003)
			};
			\addplot[color=NVGreen, mark=square*, mark size=1, thick] coordinates {
				(2, 8.400)
				(3, 8.270)
				(4, 8.120)
				(5, 8.000)
				(6, 7.920)
				(7, 7.860)
				(8, 7.810)
				(9, 7.750)
				(10, 7.680)
				(11, 7.620)
				(12, 7.570)
				(13, 7.510)
				(14, 7.460)
				(15, 7.400)
				(16, 7.360)
				(17, 7.300)
				(18, 7.270)
				(19, 7.210)
				(20, 7.150)
				(21, 7.100)
				(22, 7.060)
			};
			\node[anchor=south west] at (axis cs: 24, 0) {\includegraphics[width=.025\textwidth, trim = 210px 0px 250px 0px]{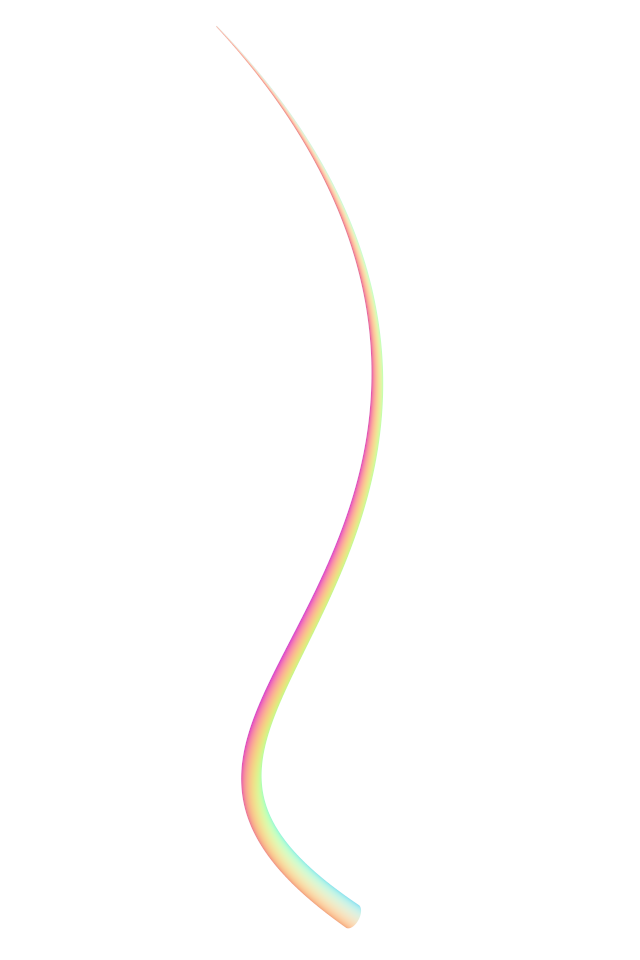}};
			\draw[densely dashed] (axis cs: 9,0) -- (axis cs: 9, 8.160);
		\end{axis}
	\end{tikzpicture}%
	&
	\begin{tikzpicture}
		\begin{axis}[
			height = .25\linewidth,
			width = .38\linewidth,
			xlabel = subdivision depth,
			ymin  = 0,
			ymax = 14,
			xmin = 1,
			xmax = 36,
			xtick ={2, 6, 9, 16, 22},
			xtick pos=left,
			ytick pos=left,
			scaled y ticks = false,
			legend style={font=\tiny,draw=none,fill=none,anchor=west,at={(0,0.82)}},
			legend cell align = {left}
		]
			\addplot[color=NVRed, mark=*, mark size=1, thick] coordinates {
			(2, 7.840)
			(3, 7.030)
			(4, 6.450)
			(5, 5.900)
			(6, 5.280)
			(7, 4.500)
			(8, 3.580)
			(9, 2.600)
			(10, 1.700)
			(11, 1.020)
			(12, 0.580)
			(13, 0.310)
			(14, 0.160)
			(15, 0.082) %
			(16, 0.040)
			(17, 0.020)
			(18, 0.010)
			(19, 0.005)
			(20, 0.002)
			(21, 0.001)
			(22, 0.0005)
			};
			\addlegendentry{Nakamaru/Ohno [12]}
			\addplot[color=NVGreen, mark=square*, mark size=1, thick] coordinates {
				(2, 6.860)
				(3, 5.970)
				(4, 5.630)
				(5, 5.400)
				(6, 5.220)
				(7, 5.060)
				(8, 4.930)
				(9, 4.790)
				(10, 4.660)
				(11, 4.530)
				(12, 4.420)
				(13, 4.300)
				(14, 4.180)
				(15, 4.060)
				(16, 3.950)
				(17, 3.830)
				(18, 3.760)
				(19, 3.660)
				(20, 3.590)
				(21, 3.540)
				(22, 3.470)
			};
			\addlegendentry{Our method}
			\node[anchor=south west] at (axis cs: 25, 0) {\includegraphics[width=.06\textwidth, trim = 100px 0 160px 0]{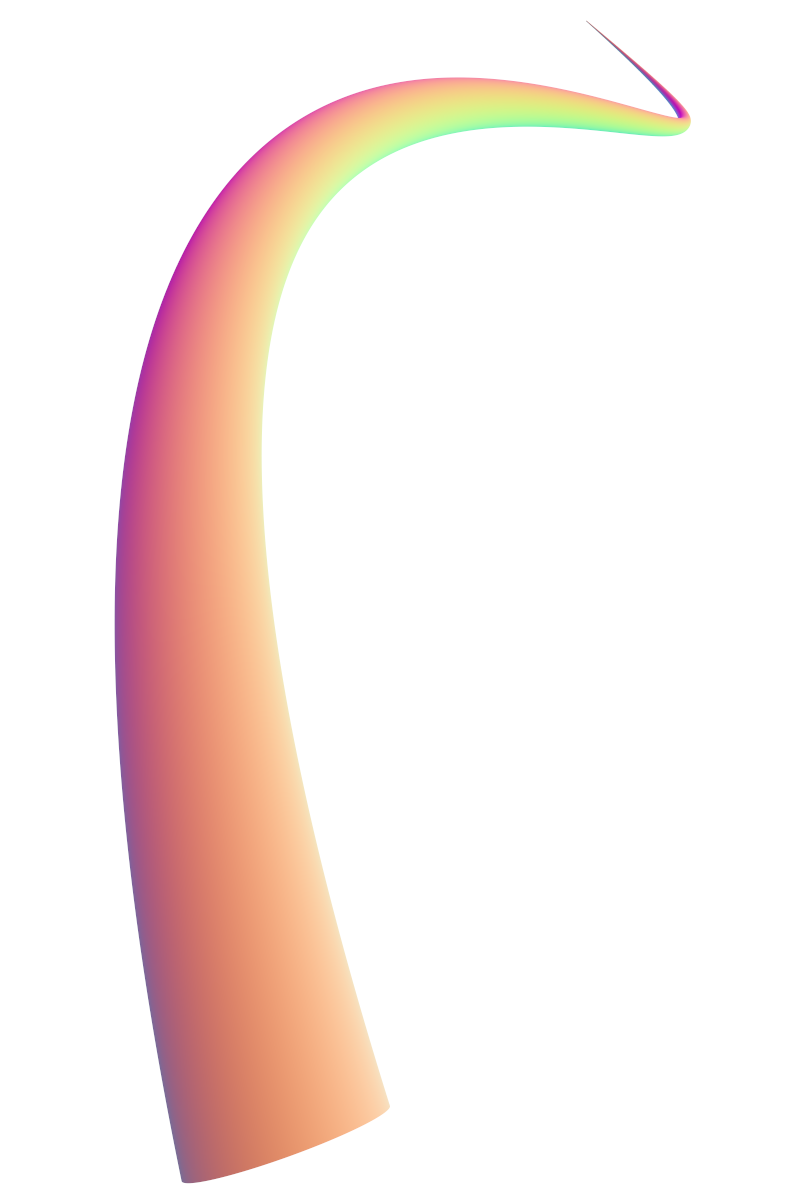}};
			\draw[densely dashed] (axis cs: 6,0) -- (axis cs: 6, 5.280);
		\end{axis}
	\end{tikzpicture}%
	\end{tabular}
	\vspace*{-4.5ex}
	\caption{The performance of methods pruning with (overlapping) AABBs drops dramatically with increasing
	subdivision depth, whereas our tight, disjoint bounding volumes always allow for efficient pruning
	and instant termination.}
	\label{fig:performance}
\end{figure}
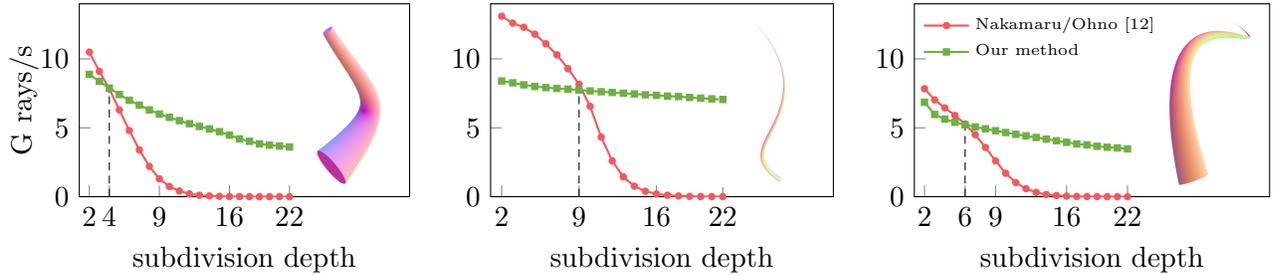

\begin{abstract}
Analyzing and identifying the shortcomings of current subdivision methods for finding intersections of rays
with fibers defined by the surface of a circular contour swept along a B\'ezier curve,
we present a new algorithm that improves precision and performance. Instead of the
inefficient pruning using overlapping axis aligned bounding boxes and determining the closest point of approach of the ray and the
curve, we prune using disjoint bounding volumes defined by cylinders and calculate the intersections on the limit
surface. %
This in turn allows for computing accurate parametric position and normal in the point of intersection.
The iteration requires only one bit per subdivision to avoid costly stack memory operations.
At a low number of subdivisions, the performance of the high precision algorithm is competitive, while for a high number of
subdivisions it dramatically outperforms the state-of-the-art.
Besides an extensive mathematical analysis, source code is provided.
\end{abstract}

\section{Introduction}

State-of-the-art photo realistic image synthesis is based on (quasi-) Monte Carlo simulation of light propagation: Rays
are traced to connect the camera with the light sources. Then, the contribution of all these light
paths is summed up. %

Often, fibers for hair and fur are part of the scenery. These fibers are usually modeled as sweep surfaces along
B\'ezier curves with a circular cross section and a parametric radius, which may vary along the curve. While triangles
are a very common representation for most other parts of the scene geometry, they often are a very unsuitable
approximation for fibers. Reasons for this exception include memory restrictions, numerical issues, and efficiency
considerations. Therefore, especially in high quality rendering, custom primitives are used for fibers, and
intersections of rays with these primitives must be found.

\section{Previous Work}

Most popular approaches are based on recursive subdivision of the curve \cite{Catmull:1974}, after which either the
segment of the curve can be approximated by a simple primitive or an iterative solver refines the solution. The number
of subdivisions required for a certain reduction of curvature can be reduced by a more thorough analysis, however at the
price of a significantly increased effort \cite{Hain:2005}.

Generalized cylinders \cite{Ballard:1982} are a common representation for hair fibers. They are defined by sweeping an
arbitrary two-dimensional contour along a three-dimensional curve. Intersections of rays with these objects can be found
without tessellation \cite{Bronsvoort:1985}: Each ray is projected into a parametric frame aligned to the trajectory
of the curve, i.e. the contour is fixed. At the same time the trajectory of the ray becomes a two-dimensional curve.
Then, the ray and the contour are subdivided simultaneously until the size of their bounding boxes fall below a
threshold. During the process, combinations of the two intervals for which the bounding boxes do not overlap can be
pruned. In a final step, the exact intersection points are calculated, which requires solving equations of higher
polynomial degree.

The method can be simplified by either restricting the shape of the sweep curve \cite{vanWijk:1984} or the shape
of the contour without subdividing the curve first \cite{vanWijk:1985}. However, finding roots of polynomials with
a high degree is still required and remains numerically challenging.

Intersections of rays and sweep surfaces with a circular cross section can also be found by combining the equations of
the trajectory of a ray and the parametric distance of a point to a parametric position on the curve
\cite{Leipelt:1995}. Again, roots of a polynomial with high degree must be found.

Approximating the intersection on the surface of the fiber with the closest point of approach of the ray and the curve
lowers the polynomial degree. The closest point of approach of two lines can be determined very efficiently in
ray-centric coordinate systems using an adaptive linearization method based on recursive subdivision
\cite{Nakamaru:2002}. If only primary visibility from a pin hole camera is of concern, it can be beneficial to compute
line samples instead of point samples \cite{Barringer:2012}. In the same spirit, cone tracing can decrease the number of
samples significantly \cite{Qin:2014}. The obtained coverage information, however, may not fit the architecture of a
fully path traced simulation.

Improvements of a ``top level'' hierarchy referencing fibers and unrolling curve subdivision such that the number of
segments matches the SIMD width may improve performance on certain architectures \cite{Woop:2014}. More recent iterative
root finding methods can replace recursive subdivision to improve convergence speed and precision at the same time
\cite{Reshetov:2017}.

While these methods computing the closest point of approach deliver state-of-the-art performance for a certain level of
detail, they all suffer from the underlying approximation, which prohibits the determination of the correct intersection
on the surface of the fiber and the normal in the intersection. An example for this issue is shown in
\Cref{fig:ribbons-vs-ours}. Furthermore, the inefficiency of the pruning tests of subdivision-based methods becomes
prohibitive for a high number of subdivisions. Finally, recursive methods using a stack suffer from memory bandwidth
limitations, especially on current GPUs.

Fast ray tracing is possible due to efficient data structures that identify all potential parts of the scene that may be
intersected by a ray. The state-of-the-art for these acceleration data structures performs hierarchical partitioning of
either space or the set of objects. Furthermore, there exist hybrid schemes that partition both the set of objects and
space in order to improve performance \cite{Stich:2009}. As the construction and traversal of such acceleration data
structures is almost orthogonal to the actual ray/fiber intersection, we focus on improving the latter in this article.

\begin{figure}
	\centering
	\begin{tabular}{@{}ccc@{}}
		\includegraphics[width=.47\textwidth]{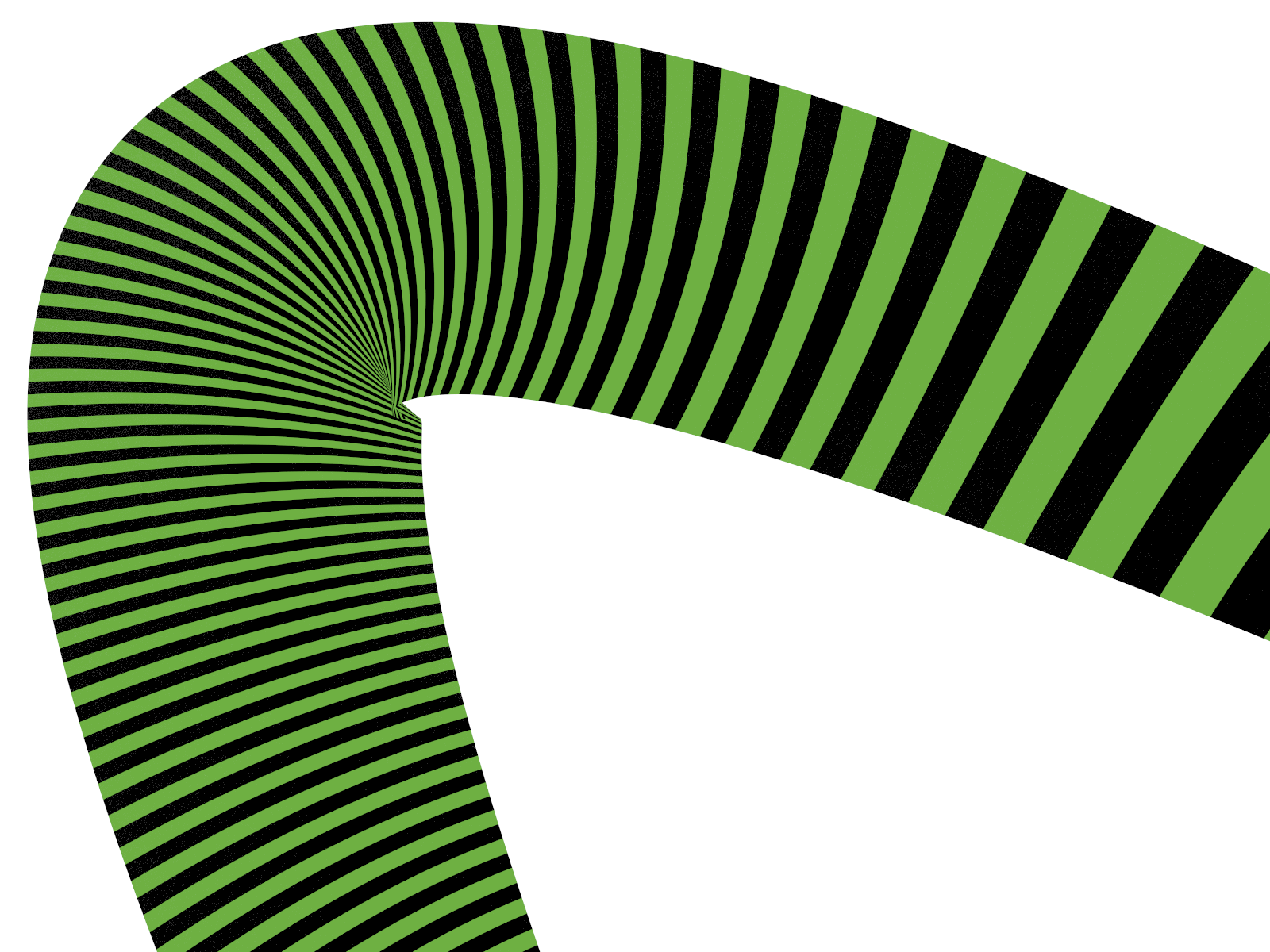}&%
		\hspace*{.02\textwidth}
		\includegraphics[width=.47\textwidth]{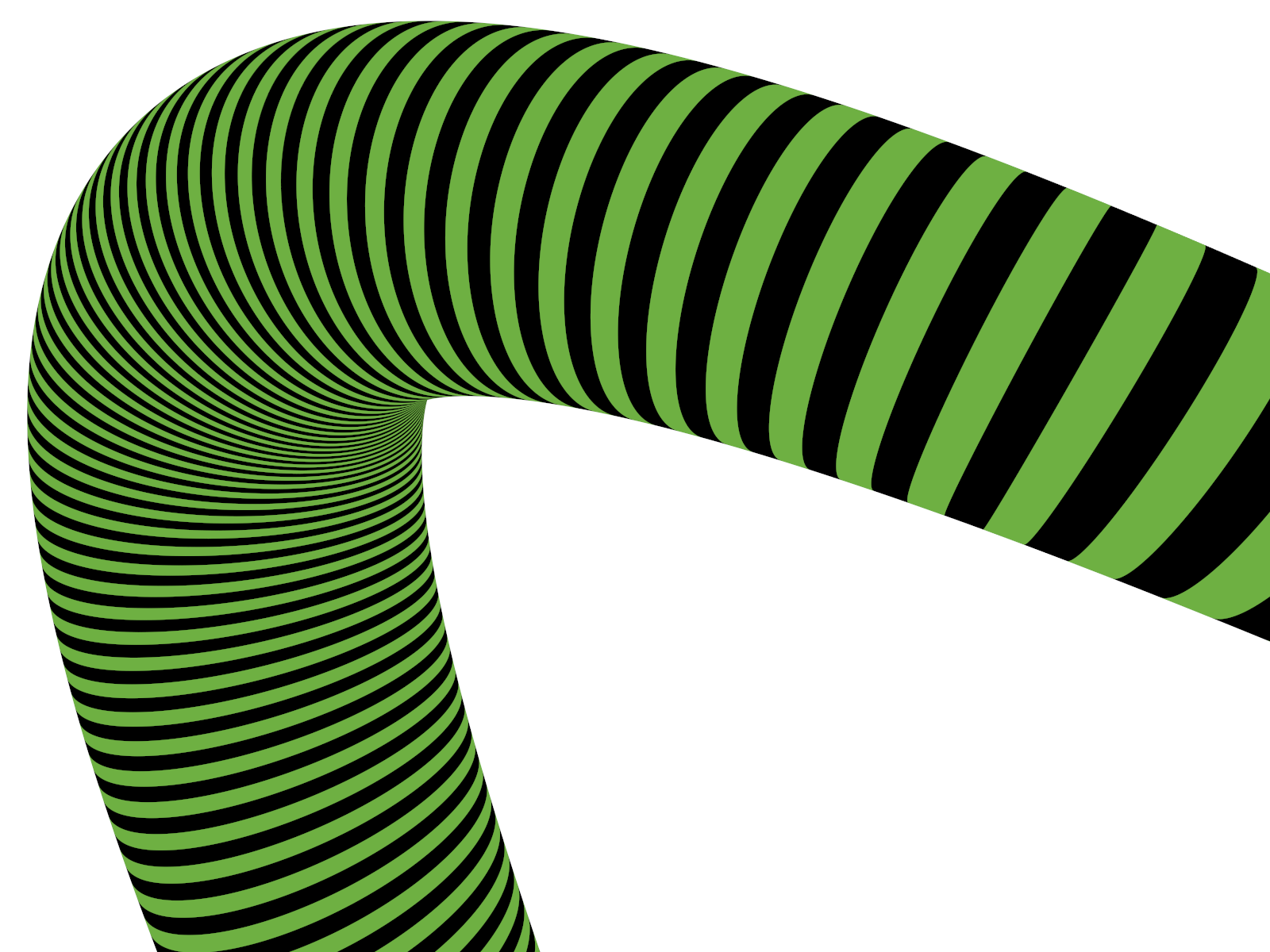}%
	\end{tabular}
	\vspace*{-2ex}
	\caption{Left: Methods based on the closest point of approach cannot determine the correct parametric position and
	normal. Right: Our method computes both with high precision.}
	\label{fig:ribbons-vs-ours}
\end{figure}

\section{Algorithm}

Our algorithm is a member of the family of subdivision-based methods computing intersections of rays with fibers by
recursively bisecting the curve and pruning regions that cannot be intersected by the ray \cite{Catmull:1974}.

The first contribution is a stackless iterative variant that only keeps track of subdivision levels that require
backtracking and re-computes all necessary data instead of employing a stack.
Our second contribution is a fast pruning test with oriented cylinders that significantly improves the accuracy,
especially for a high number of subdivisions. After a termination criterion is met, e.g. a fixed number of subdivisions,
the final intersection with the linearized segment, represented by a cylinder with oriented end caps, is computed. As
bounding cylinders of neighboring curve segments now are disjoint by construction and closer segments are always
intersected first, our algorithm can immediately terminate after an intersection has been found (3rd contribution).
Instead of approximating the actual intersection on the surface of the fiber with the closest point of approach, we
reuse the intersection already determined for pruning (4rd contribution) and are able to compute an accurate normal (5th
contribution).

\subsection{Numerically Robust Curve Representation}
\label{sec:representation}

A na\"ive implementation for cubic B\'ezier curves using four control points $(p_0, p_1, p_2, p_3)$ suffers from
severe floating point precision issues in our algorithm due to cancellation in differences required to determine the
cylinder axis ($p_3 - p_0$) and the tangent in the split point.

Therefore, we use a representation tailored to our pruning test as illustrated in \Cref{fig:representation}: We maintain the first
control point $p := p_0$, the tangents in the start and end points $t_0 := p_1 - p_0, t_1 := p_3 - p_2$, and the
direction $d := p_3 - p_0$. This representation requires different rules for the subdivision of $(p, d, t_0, t_1)$ into
$\left(p^L, d^L, t_0^L, t_1^L\right)$ and $\left(p^R, d^R, t_0^R, t_1^R\right)$, where %
\begin{align*}
	\Delta p &= \tfrac{3}{8} t_0 + \tfrac{1}{2} d - \tfrac{3}{8} t_1,\\
	t_c &= -\tfrac{1}{8} t_0 + \tfrac{1}{4} d - \tfrac{1}{8} t_1,
\end{align*}
and
\begin{align*}
	p^L   &= p,              & p^R   &= p + \Delta p,\\
	d^L   &= \Delta p,              & d^R   &= d - \Delta p,\\
	t_0^L &= \frac{1}{2}t_0, & t_0^R &= t_c,\\
	t_1^L &= t_c,            & t_1^R &= \frac{1}{2}t_1.
\end{align*}

As we will use disjoint bounding volumes,
only one of the two sets needs to be calculated as determined by the pruning test.
In fact this subdivision can be computed even slightly more efficiently than the subdivision of $(p_0, p_1, p_2,
p_3)$ and only exposes a minimal amount of instruction divergence due to branching.

\begin{figure}
\centering
\begin{minipage}{.45\textwidth}
  \centering
	\begin{tikzpicture}[scale=0.7]
		\coordinate (A) at (0, 0);
		\coordinate (B) at (1, 1.5);
		\coordinate (C) at (2, 1);
		\coordinate (D) at (4, 0);
		\node at (A) {\small \textbullet};
		\node at (B) {\small \textbullet};
		\node at (C) {\small \textbullet};
		\node at (D) {\small \textbullet};
		\draw[thick, -latex] (A) -- (B) node[midway, above, font=\small, xshift=-2] {$t_0$};
		\draw[thick, -latex] (A) -- (D) node[midway, above, font=\small] {$d$};
		\draw[thick, -latex] (C) -- (D) node[midway, above, font=\small] {$t_1$};
		\node[anchor=north, font=\small] at (A) {$p := p_0$};
		\node[anchor=south, font=\small] at (B) {$p_1$};
		\node[anchor=south, font=\small] at (C) {$p_2$};
		\node[anchor=north, font=\small] at (D) {$p_3$};
		\draw[thick, densely dotted] (A) .. controls (B) and (C)  .. (D);
	\end{tikzpicture}
  \captionof{figure}{Representing the curve with the tuple $\mathbf{\{p, d, t_0, t_1\}}$ improves the numerical robustness of the method significantly.}
	\label{fig:representation}
\end{minipage}%
\hspace*{.09\textwidth}
\begin{minipage}{.45\textwidth}
  \centering
	\begin{tikzpicture}
		\coordinate (A) at (0, 0);
		\coordinate (B) at (5, 1);
		\coordinate (C) at (-1, 1);
		\coordinate (D) at (4, 0);
		\coordinate (T) at ($-1.0*(A)-(B)+(C)+(D)$);
		\newdimen\xta
		\newdimen\yta
		\pgfextractx{\xta}{\pgfpointanchor{T}{center}}
		\pgfextracty{\yta}{\pgfpointanchor{T}{center}}
		\coordinate (N) at (-\yta, \xta);
		\coordinate (S) at ($1.0/8.0*(A)+3.0/8.0*(B)+3.0/8.0*(C)+1.0/8.0*(D)$);
		\draw[thick] (A) .. controls (B) and (C)  .. (D);
		\draw[thick, densely dotted, -latex] (0,0) -- (B);
		\draw[thick, densely dotted, latex-] (C) -- (D);
		\node at (0, 0) {\small \textbullet};
		\node[below] at (0, 0) {\small $p_0$};
		\node[below] at (B) {\small $p_1$};
		\node[below] at (C) {\small $p_2$};
		\node at (D) {\small \textbullet};
		\node[below] at (D) {\small $p_3$};
		\draw[thick, densely dashed] (S) -- ($(S)+0.35*(N)$);
		\draw[thick, densely dashed] (S) -- ($(S)-0.35*(N)$);
	\end{tikzpicture}
  \captionof{figure}{A loop violating the constraints required to construct disjoint bounding volumes must be split beforehand.}
	\label{fig:loop}
\end{minipage}
\end{figure}

\subsection{Efficient Hierarchical Pruning} %

Each region of the subdivision is conservatively bounded by an oriented cylinder, which is partitioned by a plane located in the split
point and perpendicular to the tangent in the split point of the curves. The limit surface of these cylinders guarantees that an
intersection is always mapped to the closest point on the curve.

Only if the intersection of the ray with the plane is inside the cylinder, both sub-regions must be considered. Then,
subdivision starts with the region whose bounding volume is intersected first along the ray. \Cref{fig:bounding_cylinder}
shows an example with four possible cases. Note that the test for inclusion only requires comparing the distances of the
two ray/cylinder intersections with the distance of the intersection with the plane.

As these pruning tests are performed with disjoint bounding volumes and refinement always continues with the closest
sub-region, subdivision can immediately terminate after an intersection has been found. Instant termination is essential
for a high number of subdivisions because otherwise the number of unpruned regions may grow exponentially with
subdivision depth.

While bounding both subcurves in individual cylinders instead of using one partitioned cylinder improves the culling
accuracy, the overhead of computing and intersecting two bounding volumes outweighs the theoretical benefit in practice,
especially since the benefit quickly decreases with subdivision.

\subsection{Implementation}

Pruning is performed in a ray centric coordinate system, in which the ray starts in the origin and goes along the
positive $z$ axis (``unit ray''). A reliable orthonormal basis can be efficiently constructed using Duff et al.'s recent
improvement of Frisvad's method \cite{frisvad:onb:2012, duff:onb:2017}. The transformation into the local frame only
needs to be performed once at the beginning by calculating a local set of control points. In this coordinate system, we
can simplify ray/plane intersection and the infinite cylinder intersection described by Cychosz et
al.~\cite{Cychosz:1994} significantly since $\forall v \in \mathbb{R}^3$: $\langle {v, d} \rangle = v_z$ and $v - o = v$
for a unit ray defined by its origin and direction $(o, d)$. \Cref{lst:ray-plane,lst:ray-cylinder} present
the resulting optimized intersection functions; the simplified ray/cylinder intersection is derived in
\Cref{sec:appendix_cylinder}.

We use four-dimensional control points, where the first three components are the position, and the last one defines the
radius. This consistent representation allows for cubic interpolation of the radius.

Bounding cylinders are oriented along the vector connecting the first and last control point and have a conservative
radius defined by the sum of the maximum radius in the region and the maximum distance of the inner control points to
the cylinder axis. We also use B\'ezier curves for radius interpolation, and bound the parametric radius using the
convex hull property.

An example implementation for the computation of a conservative radius for the bounding cylinders and cubic B\'ezier
curves is given in \Cref{lst:calc_radius}, using the distance of the two inner control points to the axis determined by the
method shown in \Cref{lst:distance-point-line}.

The infinite cylinders are cropped by restricting the $t$-parameter interval of the ray. After determining initial
bounds of the $t$ parameter interval, in each subdivision step one of the interval bounds is updated; both are
recalculated after backtracking. A simple implementation for cubic B\'ezier curves is shown in
\Cref{lst:recalculation}.

Recursive subdivision is performed by an iterative process by maintaining a bit string in which each subdivision level
is represented by one bit and the current size of the parametric domain. After the pruning test, the corresponding bit
in the bit string is set to one if and only if both subregions must be considered, and only in this case backtracking is required.
Then, the control points and the $t$ interval are recalculated to avoid maintaining a stack of control points. The
current parametric interval of the curve can directly be derived from the bit stack. It is always maintained in two
integer variables for start and size of the interval, which must be converted to floating point values in the unit
interval before calculating new control points. This conversion is shown in \Cref{lst:calculate_interval}.

Upon termination, the intersection of the ray with the bounding cylinder used for pruning already determines the
intersection with the linearized segment. Only very little effort is required to compute the normal and the parametric
value in the intersection, and this calculation is performed by all threads of a warp simultaneously in the very end.
\Cref{lst:calc_intersection} shows a possible implementation.

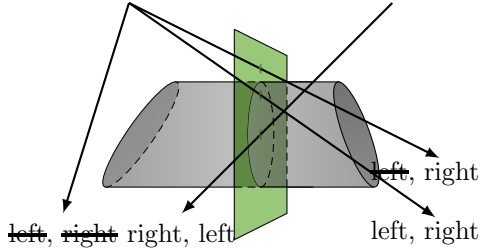
\begin{figure}
	\centering
	\begin{tikzpicture}[scale=0.35]
		\begin{scope}[rotate=-90]
		\draw[white] (-5,-2) rectangle (4,8);
			\fill[NVGreen!70!black!70] (-2, 3) arc (180:0:2cm and 0.5cm) -- (2, 4) -- (-2, 4) -- cycle;
			\fill[NVGreen!70!black!70] (0, 3) circle (2cm and 0.5cm);

			\fill[left color=gray!50!black,right color=gray!50!black,middle color=gray!50,shading=axis,opacity=0.25] (2,-2.95) -- (2,7.3) -- (-2,6) -- (-2,-0.175) -- cycle;
			\begin{scope}
				\clip[shift={(2.05, 1)}, rotate=20] (0,6) circle (2.12cm and 0.5cm);
				\fill[white] (2,0) -- (2,7.3) -- (-2,6) -- (-2,0) arc (180:360:2cm and 0.5cm);
			\end{scope}
			\begin{scope}
				\clip[shift={(-0.00, -1.55)}, rotate=-35] (0,0) circle(2.42cm and 0.5cm);
				\fill[white] (2,0) -- (2,-3.2) -- (-2,-3.2) -- (-2,0) -- cycle;
			\end{scope}
			\draw[white] (-2, 6) -- (2,7.3); %

			\fill[shift={(0.0, -1.55)}, rotate=-35, top color=black!60!,bottom color=black!40,middle color=black!50,shading=axis,opacity=0.25] (0,0) circle(2.417cm and 0.5cm);
			\begin{scope}
				\clip (-2.05, -0.2) -- (-2.05, -3.2) -- (2.05, -3.2) -- (2.05, -2.9) -- cycle;
				\draw[shift={(0.00, -1.55)}, rotate=-35] (0,0) circle(2.417cm and 0.5cm);
			\end{scope}
			\begin{scope}
				\clip (-2.05, -0.2) -- (-2.05, 1) -- (2.05, 1) -- (2.05, -2.9) -- cycle;
				\draw[densely dashed, shift={(0.00, -1.55)}, rotate=-35] (0,0) circle(2.417cm and 0.5cm);
			\end{scope}

			\begin{scope}
				\clip[shift={(0, -1.55)}, rotate=-35] (0,0) circle(2.42cm and 0.5cm);
				\fill[draw=gray!50,left color=gray!50!black,right color=gray!50!black,middle color=gray!50,shading=axis,opacity=0.15] (2,-4.9) -- (2,11.3) -- (-2,6) -- (-2,-4.2) -- cycle;
			\end{scope}

			\fill[shift={(2.05, 1)}, rotate=20,top color=black!60!,bottom color=black!40,middle color=black!50,shading=axis,opacity=0.5] (0,6) circle (2.12cm and 0.5cm);
			\draw[shift={(2.05, 1)}, rotate=20] (0,6) circle (2.12cm and 0.5cm);

			\draw (-2,6) -- (-2,-0.25);
			\draw (2,7.3) -- (2,-2.9);

			\fill[NVGreen!70] (-4, 2) -- (-3.5, 3) -- (-2, 3) arc (180:360:2cm and 0.5cm) -- (3.5, 3) -- (4, 2) -- cycle;
			\fill[NVGreen!70] (-4, 2) -- (-3, 4) -- (-2, 4) -- (-2, 2) -- cycle;
			\fill[NVGreen!70] (4, 2) -- (3, 4) -- (2, 4) -- (2, 2) -- cycle;

			\draw[densely dashed] (-2,3) arc (180:0:2cm and 0.5cm);
			\draw (-2,3) arc (180:360:2cm and 0.5cm);
			\draw (-4, 2) -- (4, 2) -- (3, 4) -- (2, 4);
			\draw[densely dashed] (2, 4) -- (-2, 4);
			\draw (-2, 4) -- (-3, 4) -- (-4, 2);

			\fill[NVGreen!70!black!90] (-2, 2) -- (-2, 3) arc (180:360:2cm and 0.5cm) -- (2, 2) -- cycle;
			\draw[densely dashed] (2, 2) -- (2, 3);
			\draw[densely dashed] (-2, 2) -- (-2, 3);

			\draw (2, 2) -- (-2, 2);
			\draw (-2, 3) -- (-2, 5);
			\draw (2, 3) -- (2, 5);

			\draw[thick](-5, -2) -- (-2.5, 3);
			\fill[NVGreen!70!black!90] (-2.5,3) circle(0.2cm and 0.05cm);
			\draw[thick, -latex, shorten >= -20] (-2.5, 3) -- (0, 8) node[pos=1.25, below, font=\small] {\soutthick{left}, right};

			\draw[thick](-5, -2) -- (-1.5, 3);
			\fill[NVGreen!50!black!90] (-1.5,3) circle(0.2cm and 0.05cm);
			\draw[thick, -latex, shorten >= -20] (-1.5, 3) -- (2, 8) node[pos=1.25, below, font=\small] {left, right};

			\draw[thick, -latex](-5, -2) -- (3, -4.5) node[pos=1.0, below, font=\small] {\soutthick{left}, \soutthick{right}};

			\draw[thick, -latex](0, 3) -- (3, 0) node[pos=1.0, below, font=\small] {right, left};
			\fill[NVGreen!50!black!90] (0,3) circle(0.2cm and 0.05cm);
			\draw[thick](-5, 8) -- (0, 3);
		\end{scope}
	\end{tikzpicture}
	\vspace*{-2ex}
	\caption{Bounding cylinders with partitioning planes improve the efficiency of pruning tests and allow for instant
termination after an intersection has been found.}
	\label{fig:bounding_cylinder}
	\vspace*{-0.5ex}
\end{figure}

\subsection{Constraints and Limitations}

A bisection into subregions that can be bounded by disjoint bounding volumes poses well-defined restrictions on the
allowed sets of control points. At the same time,
the constraints also ensure that curves with valid configurations cannot be split into invalid ones. For cubic B\'ezier
curves the constraints
\begin{align*}
	\langle {p_2 - p_0, p_1 - p_0} \rangle &\geq 0,\\
	\langle {p_3 - p_1, p_1 - p_0} \rangle &\geq 0,\\
	\langle {p_3 - p_1, p_3 - p_2} \rangle &\geq 0,\\
	\langle {p_2 - p_0, p_3 - p_2} \rangle &\geq 0,\\
	\langle {p_2 - p_0, p_3 - p_1} \rangle &\geq 0
\end{align*}
guarantee that the curve can be recursively split into subcurves with disjoint bounding volumes.
\Cref{sec:appendix_constraints} provides the constraints required for quadratic B\'ezier curves and all necessary
proofs for both quadratic and cubic B\'ezier curves. Curves that do not fulfill these constraints must be subdivided beforehand. An example for such a configuration is shown in
\Cref{fig:loop}.

While these constraints are necessary, they are not sufficient to guarantee disjoint bounding volumes of fibers: If a
point on the surface (perpendicular to the tangent $t_u$ of the point on the curve $p_u$, at a distance defined by the
radius in that point $r_u$) intersects the partitioning plane in the split point, a valid part of the surface of the
fiber will be cropped. %

\begin{figure}
	\centering
	\includegraphics[width=.5\textwidth]{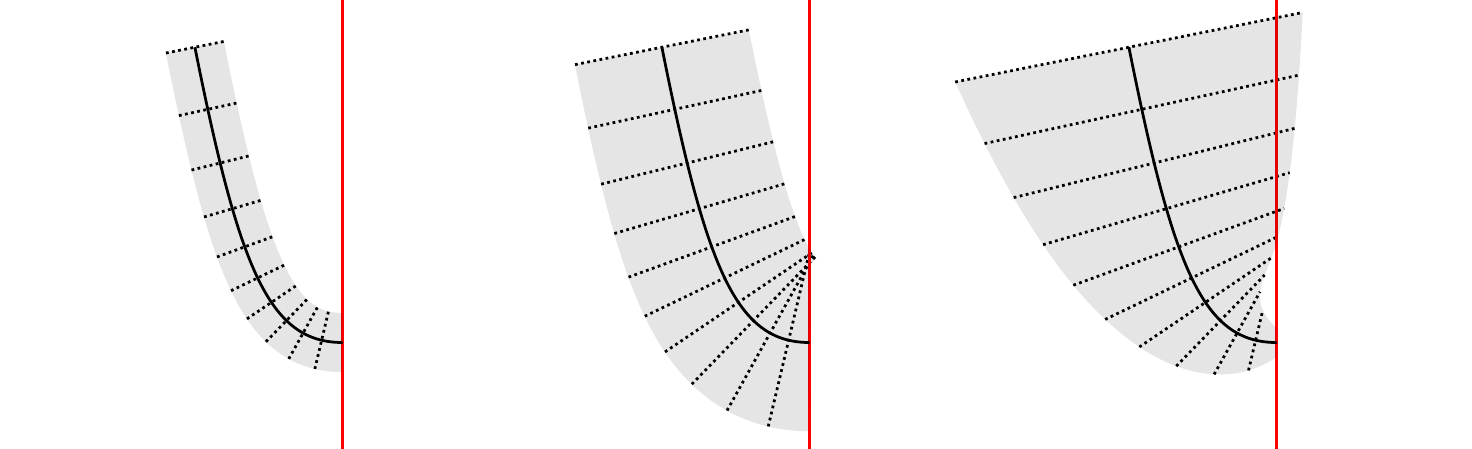}
	\caption{While curves may be split disjointly in cusps, thicker fibers can still overlap splitting planes, causing
	visible defects. While the left fiber does not suffer from this issue, the fiber in the middle, which goes along the
	same curve but has a larger radius does. Of course, fibers with parametric radius (right) can also overlap
	partitioning planes.}
	\label{fig:issue-cusps}
\end{figure}

\Cref{fig:issue-cusps} shows three examples: The leftmost fiber with a small, constant radius does not suffer from this
issue, while the surface of the other ones intersects the split plane as their radius is too large, and thus errors may
be introduced. Besides thick fibers, high curvatures (e.g. in cusps) can cause similar issues. Regions with such a
behavior must be isolated and require subdivision beforehand, too.
Note that a similar issue affects methods based on the closest point of approach.

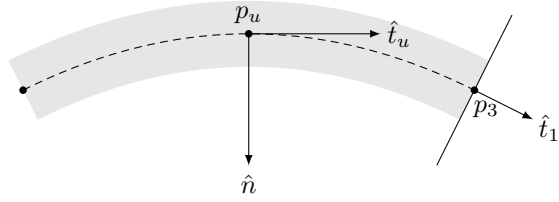
\begin{figure}%
	\centering
	\begin{tikzpicture}
		\coordinate (A) at (0, 0);
		\coordinate (B) at (2, 1);
		\coordinate (C) at (4, 1);
		\coordinate (D) at (6, 0);
		\draw[line width=25, opacity=0.1] (A) .. controls (B) and (C) .. (D);
		\draw[densely dashed] (A) .. controls (B) and (C) .. (D);
		\draw (5.5, -1) -- (6.5, 1);
		\draw[-latex, shorten >= 7] (D) -- (7, -0.5) node[font=\small] {$\hat t_1$};
		\draw[-latex, shorten >= 7] (3, 0.75) -- (5, 0.75) node[font=\small] {$\hat t_u$};
		\draw[-latex, shorten >= 7] (3, 0.75) -- (3, -1.25) node[font=\small] {$\hat n$};
		\node[anchor = south, font=\small] at (3, 0.75) {$p_u$};
		\fill (3, 0.75) circle (0.05);
		\fill (A) circle (0.05);
		\fill (D) circle (0.05);
		\node[anchor=north west, font=\small, xshift=-4] at (D) {$p_3$};
	\end{tikzpicture}
	\caption{Illustration of finding the closest point on the normal plane in a point on the curve to a partitioning
plane. \label{Fig:EvilConfiguration}} %
\end{figure}

Such configurations can be identified as shown in \Cref{Fig:EvilConfiguration}:
The surface point closest to the plane in the split point perpendicular to the split tangent $t_1$ is on the plane
perpendicular to the tangent $t$ at the parametric position. Gram-Schmidt orthogonalization yields the displacement
from the curve closest to the split plane
\begin{align*}
n_u &= t_1 - \langle {t_1, \hat t_u} \rangle \cdot \hat t_u\\
    &= t_1 - \frac{\langle {t_1, t_u \rangle}t_u}{\langle {t_u, t_u} \rangle}\\
    &\propto \langle {t_u, t_u} \rangle t_1 - \langle {t_1, t_u \rangle}t_u{},
\end{align*}
where $t_u = (1-u)^2(p_1 - p_0) + 2(1-u)u(p_2 - p_1) + u^2(p_3 - p_2)$ and $t_1 = p_3 - p_2$ for a cubic B\'ezier curve.
As
$n_u$ has a maximum degree of 4, checking all solutions $u_i \in [0, 1]$ of
\begin{align*}
\langle {p_3 - p_u + r_u \cdot \hat n_u, \hat t_1} \rangle &= 0&\textrm{or just}\\
\langle {p_3 - p_u + \bar r \cdot \hat n_u, \hat t_1} \rangle &= 0,&r_u \leq \bar r
\end{align*}
requires solving a quartic equation, e.g. using Ferrari's method \cite{Smith:1929}, eliminating the cubic term using
the Tschirnhaus transformation \cite{Boyer:1968}.

As the curvature decreases with subdivision (%
see subdivision rules in \Cref{sec:representation}%
) and
the maximum radius $\bar r$ is a constant, it is sufficient to check each fiber only initially in its two end points
with corresponding tangents. %

While it would be trivially possible to support all possible curve configurations and overcome the issues in cusps and
with very thick fibers by optionally allowing overlapping bounding volumes in certain subdivisions, the overhead of
checking the criteria in every subdivision step and recording the ones that overlap will most likely not pay off in
practice since the number of such configurations is usually small. As furthermore the cost of
splitting beforehand is rather moderate, the overall penalty tends to be negligible.

\section{Results and Discussion}

We evaluated the performance and precision of the presented algorithm using single fibers along quadratic and cubic
B\'ezier curves.

\begin{figure}
	\centering
	\begin{tabular}{@{}ccc@{}}
		\includegraphics[width=.21\textwidth]{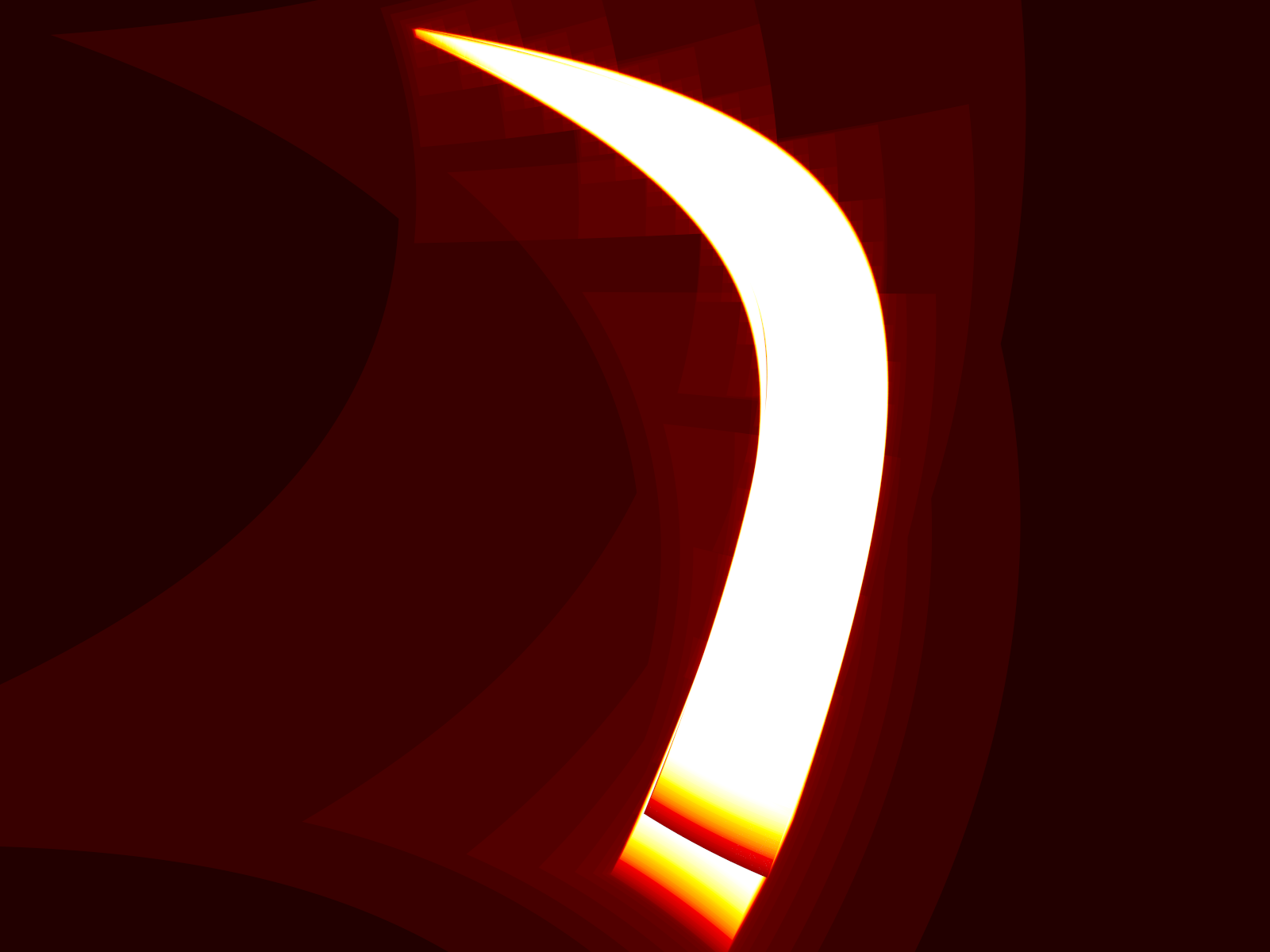}&%
		\includegraphics[width=.21\textwidth]{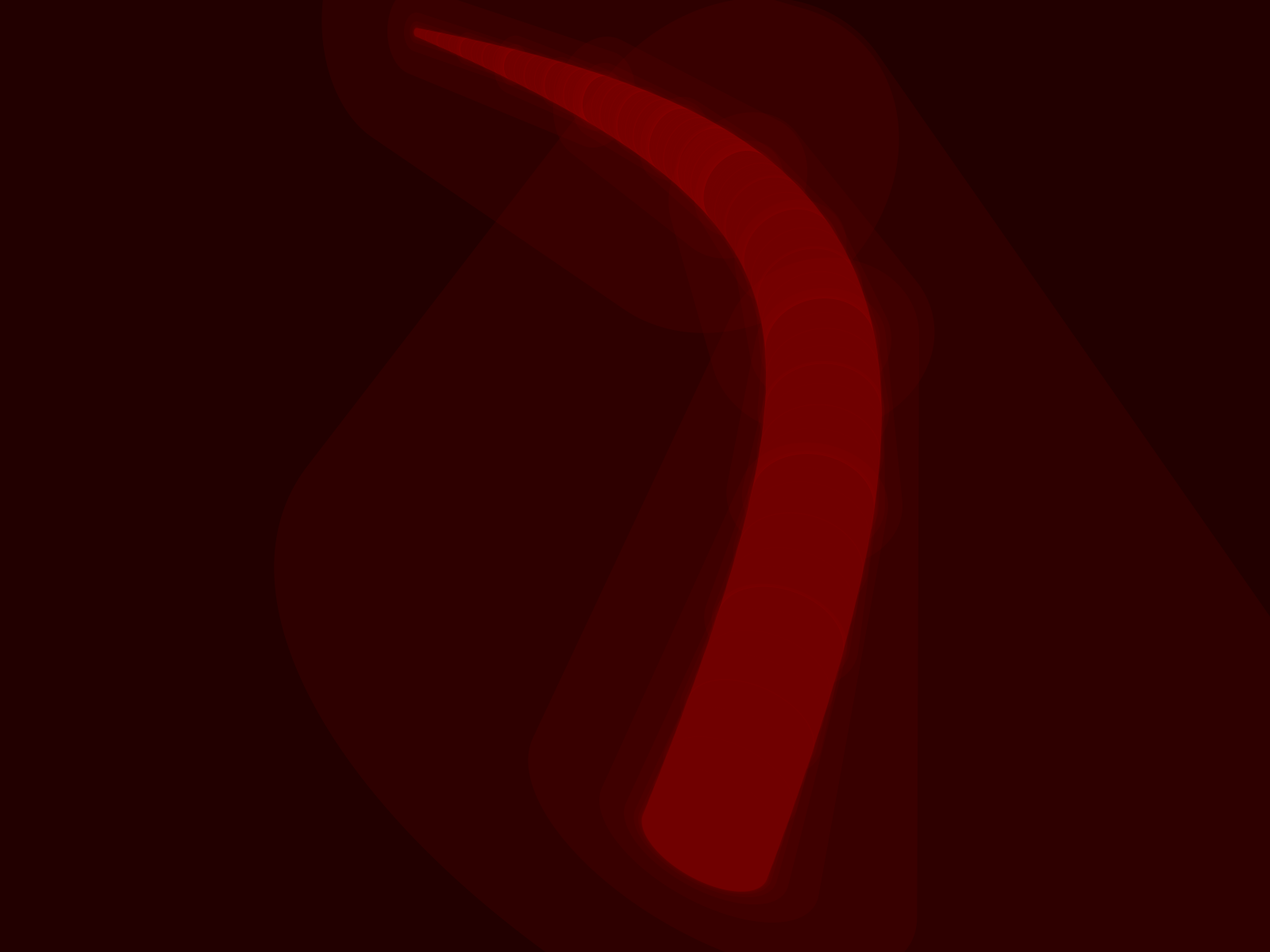}&%
		\begin{tikzpicture}[scale=0.233]
		\clip (-0.03, -0.03) rectangle (3, 12.03);
		\pgfdeclareverticalshading{flame} {3cm}{rgb(0cm)=(0,0,0); rgb(0.5cm)=(0,0,0); rgb(1cm)=(1,0,0); rgb(2cm)=(1,1,0); rgb(2.5cm)=(1,1,1); rgb(3cm)=(1,1,1)}
		\draw[shading=flame](0,0) rectangle (1,12);
		\node[anchor=south, font=\tiny, yshift=-3] at (1.6,0){0};
		\node[font=\tiny] at (1.6,4){15};
		\node[font=\tiny] at (1.6,8){30};
		\node[anchor=north, font=\tiny, yshift=3] at (1.6,12){75};
		\end{tikzpicture}%
	\end{tabular}
	\vspace*{-2ex}
	\caption{Pruning with oriented cylinders (right) dramatically lowers the total number of regions of interest by reducing the
number false positives compared to pruning with axis aligned bounding boxes (left), especially for a high number of
subdivisions.}
	\label{fig:heatmap-pruning}
\end{figure}
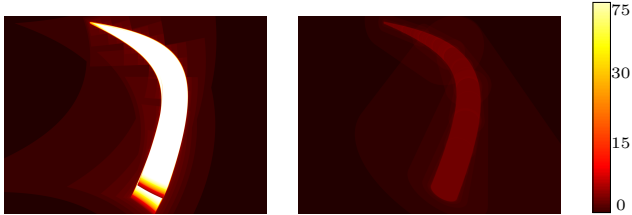

\Cref{fig:heatmap-pruning} compares the number of pruning tests required for our method to pruning with axis
aligned bounding boxes for a high number of subdivisions. As expected the overlapping axis aligned bounding boxes result
in an exponential growth of regions that cannot be pruned. Note that the bounding boxes are aligned to the main axis of the
coordinate system of each ray, hence they appear to be warped.

\Cref{fig:performance} shows the performance for intersecting $\sim$1 million rays with a single fiber on an NVIDIA
Titan V. While the relative performance does depend on the orientation and curvature of the fiber, as pruning with axis
aligned bounding boxes benefits from straight, aligned fibers, the main issue for the dramatic slowdown of the
state-of-the art remains: After a certain number of subdivisions, the bounding boxes of two curves resulting from the
subdivision of a region overlap almost entirely. Then, none of them can be pruned and a significant amount of additional
backtracking is required.

Nakamaru's method \cite{Nakamaru:2002} does not take into account the ray direction for deciding which subcurve is
checked first, therefore causing exponential growth of valid regions in the worst case. Hybrids between the two methods,
i.e. pruning with axis aligned bounding boxes, taking ray direction into account, and calculating the final intersection
with cylinders not only suffer from divergence caused by the additional test, but are mostly limited by the pruning
inefficiency, and therefore may be only valuable in rare cases.

Combining the fiber intersection with a top level hierarchy referencing fibers, or regions on fibers if
the fiber must be split beforehand, is straightforward. The top level hierarchy is primarily orthogonal to fiber intersection
unless fibers are partitioned into very small regions. Then, the increased efficiency of the pruning tests of the
presented method becomes even more important.

Note that in practice, intersection cost is almost always dominated by the performance of traversing the hierarchy
referencing the fibers and other geometry in the scene. Nevertheless, we calculate accurate intersections and normals on the
surface in a reliable way and do so either with a small overhead or -- if high precision is of concern -- dramatically
faster.

\section{Conclusion}

We have presented an algorithm that outperforms the state-of-the-art subdivision-based ray/fiber intersection method
significantly for a high number of subdivisions, and computes accurate intersections on the surface of the fiber with an
accurate normal. In addition, we also determine a precise parametric position. Even for small numbers of subdivisions,
which quite often lead to visible artifacts in Nakamaru's method \cite{Nakamaru:2002}, the approximation error of our
algorithm is well understood and the overhead of pruning with oriented cylinders instead of axis aligned bounding boxes
remains reasonable.
While the algorithm cannot handle arbitrary fibers, configurations that cannot be supported can be identified in
advance. Subdividing such fibers resolves the issues.

Future opportunities include displacements and arbitrary contours, both only requiring an additional intersection test
after pruning with conservative bounding cylinders.

\newcommand{\etalchar}[1]{$^{#1}$}

\newpage
\appendix

\section{Intersection of a unit ray with an infinite cylinder} \label{sec:appendix_cylinder}

The ray/cylinder intersection described by Cychosz et al \cite{Cychosz:1994} can be simplified for unit rays,
because a normalization of the cylinder axis is not required:

The smallest distance between the ray and an infinite line through the points $p$ and $q$, i.e. with axis $A := q - p$,
normalized to $\hat A := \frac{A}{|A|}$ is
\begin{align*}
d &= |p \cdot D|\textrm{, where}\\
D &= \frac{\hat R \times \hat A}{|\hat R \times \hat A|}.
\end{align*}
For $\bar{l}_A := \frac{1}{|A|}$,
\begin{align*}
\hat R \times \hat A &= \hat R \times (A \cdot \bar{l}_A) = (\hat R \times A) \cdot \bar{l}_A %
= (-a_y, a_x, 0)^T \cdot \bar{l}_A.
\end{align*}
Furthermore,
\begin{align*}
|\hat R \times \hat A| &= \sqrt{\bar{l}_A^2(a_y^2 + a_x^2)} = \bar{l}_A\sqrt{a_x^2 + a_y^2}.
\end{align*}
Therefore,
\begin{align*}
D &= \frac{\hat R \times \hat A}{|\hat R \times \hat A|}
= \frac{(-a_y, a_x, 0)^T \cdot \bar{l}_A}{\sqrt{a_x^2 + a_y^2} \cdot \bar{l}_A}
= \frac{(-a_y, a_x, 0)^T}{\sqrt{a_x^2 + a_y^2}},
\end{align*}
and the \textbf{squared} minimum distance between
the ray and the line is
\begin{align}
d^2 &= \frac{(a_xp_y - a_yp_x)^2}{a_x^2 + a_y^2} \label{eq:sqr_ray_line_distance}.
\end{align}

The ray hits the infinite cylinder though $p$ and $q$ with the radius $r$ if and only if
\begin{align*}
d^2 &\leq r^2.
\end{align*}
For the special cases of a line along the $z$ axis, in which the denominator of (\ref{eq:sqr_ray_line_distance}) would
be zero, the squared distance between the ray and the line is the two-dimensional distance of $p$ to the origin, i.e.
$d^2 = p_x^2 + p_y^2$.

Setting $g := a_x^2 + a_y^2$, the distance of the closest point of approach (cpa) of the ray and the cylinder from the ray origin is
\begin{align*}
t_{cpa} &= \frac{(p \times \hat A) \cdot D} {| \hat R \times \hat A |}\nonumber\\
&= \frac{\bar{l}_A(p \times A) \cdot \frac{(-a_y, a_x, 0)^T}{\sqrt{g}}} {\bar{l}_A\sqrt{g}}\nonumber\\
&= \frac{(p \times A) \cdot (-a_y, a_x, 0)^T}{g}\nonumber\\
&= \frac{(p_ya_z - p_za_y, p_za_x - p_xa_z, p_xa_y - p_ya_x)^T \cdot (-a_y, a_x, 0)^T}{g}\nonumber\\
&= \frac{-p_ya_ya_z + p_za_y^2 + p_za_x^2 - p_xa_za_x}{g}\nonumber\\
&= \frac{p_zg - a_z(p_ya_y + p_xa_x)}{g}\nonumber\\
&= p_z - \frac{a_z(p_ya_y + p_xa_x)}{g}.
\end{align*}

The intersections of the ray and the cylinder are located at $(0, 0, t \pm s)$, where
\begin{align*}
s &= \left|\frac{\sqrt{r^2 - d^2}}{\hat R \cdot \hat O}\right|&\textrm{and}\\
\hat O &= \frac{D \times \hat A}{|D \times \hat A|}.
\end{align*}

Again, a normalization of $A$ is not required, as its inverse length $\bar{l}_A$ cancels out:
\begin{align*}
D \times \hat A &= \frac{(-a_y, a_x, 0)^T}{\sqrt{g}} \times (A \cdot \bar{l}_A) = \left((-a_y, a_x, 0)^T \times A\right) \cdot \frac{\bar{l}_A}{\sqrt{g}}\\
&= (a_xa_z, a_ya_z, -g)^T \cdot \frac{\bar{l}_A}{\sqrt{g}},\\
\end{align*}
Hence,
\begin{align*}
|D \times \hat A| &= \sqrt{\frac{\bar{l}_A^2}{\sqrt{g}^2}(a_x^2a_z^2+a_y^2a_z^2 + g^2)} = \frac{\bar{l}_A}{\sqrt{g}} \sqrt{a_z^2(a_x^2 + a_y^2) + g^2} \nonumber \\
&= \frac{\bar{l}_A}{\sqrt{g}} \sqrt{a_z^2g + g^2}, \nonumber \\
\hat O &= \frac{(a_xa_z, a_ya_z, -g)^T}{\sqrt{a_z^2g + g^2}}
= \frac{(a_xa_z, a_ya_z, -g)^T}{\sqrt{g(a_z^2 + g)}}.
\end{align*}

For a unit ray with $\hat R = (0, 0, 1)^T$, we only need to consider
\begin{align*}
O_z &= -\frac{g}{\sqrt{g(a_z^2 + g)}}
= -\sqrt{\frac{g}{a_z^2 + g}},
\end{align*}
since all terms are positive.
Finally,
\begin{align*}
s &= \frac{\sqrt{r^2 - d^2}}{\sqrt{\frac{g}{a_z^2 + g}}}
= \sqrt{\frac{r^2 - d^2}{\frac{g}{a_z^2 + g}}}
= \sqrt{\frac{(r^2 - d^2)(a_z^2+g)}{g}},
\end{align*}
and the distance to the closest visible surface boundary is
\begin{align*}
t = \begin{cases}
t_{cpa} - s & t_{cpa} \geq s\\
t_{cpa} + s & \textrm{otherwise.}
\end{cases}
\end{align*}
If $t < 0$, the cylinder is behind the ray origin.

\onecolumn
\section{Curve Constraints}

\label{sec:appendix_constraints}
\subsection{Quadratic B\'ezier Curves}
\begin{lemma}
	The two subcurves defined by their control points $\left(p^L_0, p^L_1, p^L_2\right)$ and $\left(p^R_0, p^R_1, p^R_2\right)$
	resulting from de Casteljau subdivision \cite{Casteljau:1959} in the domain center of the quadratic B\'ezier curve defined by
	the control points $(p_0, p_1, p_2)$ can be enclosed by disjoint bounding volumes partitioned by a plane located in
	the split point $s := \frac{p_0 + 2p_1 + p_2}{4}$ orthogonal to the tangent in the split point $t := p_2 - p_0$ if and only if
	\begin{align}
		\langle {p_1 - p_0, p_1 - p_2} \rangle &\leq 0.\label{eq:constraint_quadratic_bezier}
	\end{align}
\end{lemma}

\begin{proof}
	As B\'ezier curves are defined as a convex combination of the control points, the bounding volumes of the two sub
curves are disjointedly can be split by the plane located in $s$ and orthogonal to $t$ if
	\begin{align}
		\left\langle p^L_i - s, t \right\rangle &\leq 0 &\textrm{and}\\
		\left\langle p^R_i - s, t \right\rangle &\geq 0 &\forall i \in \lbrace 0, 1, 2 \rbrace .
	\end{align}

	De Casteljau subdivision in the domain center creates the control points
	\begin{align*}
		\left(p^L_0, p^L_1, p^L_2\right) &:= \left(p_0, \frac{p_0 + p_1}{2}, \frac{p_0 + 2 p_1 + p_2}{4}\right)&\textrm{and}\\
		\left(p^R_0, p^R_1, p^R_2\right) &:= \left(\frac{p_0 + 2 p_1 + p_2}{4}, \frac{p_1 + p_2}{2}, p_2\right)
	\end{align*}
	of the two subcurves.

	$p^L_2 = p^R_0 = s$ obviously fulfill the conditions by construction.

	For $p^L_0 := p_0$
	\begin{alignat}{2}
         && \left\langle {p_0 - \frac{p_0 + 2p_1 + p_2}{4}, p_2 - p_0} \right\rangle                                           &\leq 0\nonumber\\
		\LRA && \langle {3p_0 - 2p_1 - p_2, p_2 - p_0} \rangle                                                          &\leq 0\nonumber\\
		\LRA && \underbrace{\langle {p_0 - p_2, p_2 - p_0} \rangle}_{\leq 0} + \langle {2p_0 - 2p_1, p_2 - p_0} \rangle   &\leq 0\nonumber\\
		\LA  && \langle {p_0 - p_1, p_2 - p_0} \rangle                                                                  &\leq 0.
	\end{alignat}

	Analogously, checking $p^R_2 := p_2$ gives $\langle {p_1 - p_2, p_2 - p_0} \rangle \leq 0$. The reversed conditions
	\begin{align}
	\langle {p_1 - p_0, p_2 - p_0} \rangle &\geq 0 &\textrm{and}\label{eq:constraint_quadratic_bezier_subcurves_1}\\
	\langle {p_1 - p_2, p_0 - p_2} \rangle &\geq 0\label{eq:constraint_quadratic_bezier_subcurves_2}
	\end{align}
	can be combined to
	\begin{align}
		\langle {p_1 - p_0, p_1 - p_2} \rangle &\leq 0,
	\end{align}
	as the projections of $p_1 - p_0$ and $p_1 - p_2$ onto $p_2 - p_0$ must have different signs to satisfy both conditions (so that
	$p_1$ is inside the hatched area in \Cref{fig:quadratic_bezier_subcurves}).

	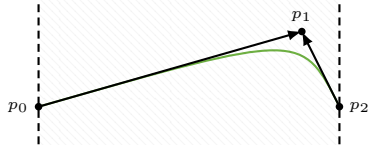
\begin{figure}[h]
		\centering
		\begin{tikzpicture}
			\coordinate (A) at (0, 0);
			\coordinate (B) at (3.5, 1);
			\coordinate (C) at (4, 0);
			\coordinate (T) at ($-1.0*(A)+(C)$);
			\newdimen\xta
			\newdimen\yta
			\pgfextractx{\xta}{\pgfpointanchor{T}{center}}
			\pgfextracty{\yta}{\pgfpointanchor{T}{center}}
			\coordinate (N) at (-\yta, \xta);
			\coordinate (S) at ($1.0/8.0*(A)+6.0/8.0*(B)+1.0/8.0*(C)$);
			\draw[HighlightColor, thick] (A) .. controls (B) and (B)  .. (C);
			\draw[thick, -latex] (A) -- (B);
			\draw[thick, -latex] (C) -- (B);
			\fill[black] (A) circle (0.05);
			\fill[black] (B) circle (0.05);
			\fill[black] (C) circle (0.05);
			\node[left] at (0, 0) {\tiny $p_0$};
			\node[above] at (B) {\tiny $p_1$};
			\node[right] at (C) {\tiny $p_2$};
			\draw[thick, densely dashed] ($(A) - (0, 0.5)$) -- ($(A) + (0, 1.5)$);
			\draw[thick, densely dashed] ($(C) - (0, 0.5)$) -- ($(C) + (0, 1.5)$);
			\fill[pattern = north west lines, opacity=0.2] ($(A) - (0, 0.5)$) rectangle ($(C) + (0, 1.5)$);
		\end{tikzpicture}
		\caption{$p_1$ must be inside the hatched area if and only if (\ref{eq:constraint_quadratic_bezier_subcurves_1}) and
(\ref{eq:constraint_quadratic_bezier_subcurves_2}) are satisfied. Only in that case
(\ref{eq:constraint_quadratic_bezier}) is also fulfilled.}
		\label{fig:quadratic_bezier_subcurves}
	\end{figure}

	Finally, for $p^L_1 := \frac{1}{2}p_0 + \frac{1}{2}p_1$ the conditions are always met, as
	\begin{alignat}{2}
         && \left\langle {p^L_1 - s, t} \right\rangle                                               &\leq 0\nonumber\\
		\LRA && \left\langle {\frac{p_0 + p_1}{2} - \frac{p_0 + 2p_1 + p_2}{4}, p_2 - p_0} \right\rangle &\leq 0\nonumber\\
		\LRA && \langle {2p_0 + 2p_1 - (p_0 + 2p_1 + p_2), p_2 - p_0} \rangle                 &\leq 0\nonumber\\
		\LRA && %
		\langle {p_0 - p_2, p_2 - p_0} \rangle%
		&\leq 0.
	\end{alignat}
	The remaining proof for $p^R_1$ is analogous. As all control points are now on one side of the splitting plane, their
convex hull property ensures that all control points of further subcurves and the subcurve itself must also be on this
side.
\end{proof}

\begin{lemma}
	Splitting the quadratic B\'ezier curve fulfilling (\ref{eq:constraint_quadratic_bezier}) in the domain center results
in two subcurves that both also satisfy (\ref{eq:constraint_quadratic_bezier}). Hence, all nested subcurves fulfill
the condition.
\end{lemma}

\begin{proof}
For the left subcurve
	\begin{alignat}{2}
				 && \left\langle {p^L_1 - p^L_0, p^L_1 - p^L_2} \right\rangle                                                &\leq 0\nonumber\\
		\LRA && \left\langle {\frac{p_0 + p_1}{2} - p_0, \frac{p_0 + p_1}{2} - \frac{p_0 + 2p_1 + p_2}{4}} \right\rangle &\leq 0\nonumber\\
		\LRA && \left\langle {p_1 - p_0, 2p_0 + 2p_1 - p_0 - 2p_1 - p_2} \right\rangle                                   &\leq 0\nonumber\\
		\LRA && \underbrace{\langle {p_1 - p_0, p_0 - p_2} \rangle}_{\geq 0\textrm{ by Eq.~(\ref{eq:constraint_quadratic_bezier_subcurves_1})}} &\leq 0,\nonumber%
	\end{alignat}
	and for the right subcurve
	\begin{alignat}{2}
				 && \left\langle {p^R_1 - p^R_0, p^R_1 - p^R_2} \right\rangle                                                &\leq 0\nonumber\\
		\LRA && \left\langle {\frac{p_1 + p_2}{2} - \frac{p_0 + 2p_1 + p_2}{4}, \frac{p_1 + p_2}{2} - p_2} \right\rangle &\leq 0\nonumber\\
		\LRA && \langle {2p_1 + 2p_2 - (p_0 + 2p_1 + p_2), p_1 + p_2 - 2p_2} \rangle                          &\leq 0\nonumber\\
		\LRA && \langle {p_2 - p_0, p_1 - p_2} \rangle                                                       &\leq 0\nonumber\\
		\LRA && \underbrace{\langle {p_2 - p_1, p_0 - p_2} \rangle}_{\geq 0\textrm{ by Eq.~(\ref{eq:constraint_quadratic_bezier_subcurves_2})}} &\leq 0.\nonumber%
	\end{alignat}
\end{proof}

\subsection{Cubic B\'ezier Curves}
\begin{lemma}
	The two subcurves defined by their control points $\left(p^L_0, p^L_1, p^L_2, p^L_3\right)$ and $\left(p^R_0, p^R_1, p^R_2, p^R_3\right)$
	resulting from de Cateljau subdivision \cite{Casteljau:1959} in the domain center of a cubic B\'ezier curve defined by
	the control points $(p_0, p_1, p_2, p_3)$ can be enclosed in disjoint bounding volumes partitioned by a plane located in
	the split point $s := \frac{1}{8}(p_0 + 3p_1 + 3p_2 + p_3)$ and orthogonal to the tangent in the split point $t :=
	\frac{3}{4}(-p_0 - p_1 + p_2 + p_3)$ if
	\begin{align}
		\langle {p_2 - p_0, p_1 - p_0} \rangle &\geq 0 & \qquad\textrm{and}\label{eq:cubic_bezier_p2p0p1p0}\\
		\langle {p_3 - p_1, p_1 - p_0} \rangle &\geq 0 & \qquad\textrm{and}\label{eq:cubic_bezier_p3p1p1p0}\\
		\langle {p_3 - p_1, p_3 - p_2} \rangle &\geq 0 & \qquad\textrm{and}\label{eq:cubic_bezier_p3p1p3p2}\\
		\langle {p_2 - p_0, p_3 - p_2} \rangle &\geq 0 & \qquad\textrm{and}\label{eq:cubic_bezier_p2p0p3p2}\\
		\langle {p_2 - p_0, p_3 - p_1} \rangle &\geq 0.\label{eq:cubic_bezier_p2p0p3p1}
	\end{align}
\end{lemma}
\begin{proof}
	As a B\'ezier curve results from a convex combination of its control points, the bounding volumes of the two sub
curves can disjointedly be split by the plane located in $s$ and orthogonal to $t$ if
	\begin{align*}
		\left\langle p^L_i - s, t \right\rangle &\leq 0 &\textrm{and}\\
		\left\langle p^R_i - s, t \right\rangle &\geq 0 &\forall i \in \lbrace 0, 1, 2, 3 \rbrace .
	\end{align*}

	De Casteljau subdivision in the domain center creates the control points
	\begin{align*}
		\left(p^L_0, p^L_1, p^L_2, p^L_3\right) &:= \left(p_0, \frac{p_0 + p_1}{2}, \frac{p_0 + 2 p_1 + p_2}{4}, \frac{p_0 + 3p_1 + 3p_2 + p_3}{8}\right)&\textrm{and}\\
		\left(p^R_0, p^R_1, p^R_2, p^R_3\right) &:= \left(\frac{p_0 + 3p_1 + 3p_2 + p_3}{8}, \frac{p_1 + 2 p_2 + p_3}{4}, \frac{p_2 + p_3}{2}, p_3\right)
	\end{align*}
	of the two subcurves.

	Obviously $p^L_3 = p^R_0 = s$ meet the conditions by construction with the given constraints\footnote{Note that the
curves do touch in $s$ and hence their bounding volumes would have a shared plane. In practice, we still work with
disjoint bounding volumes by assigning everything on the partitioning plane explicitly to the subcurve tested first; the
result obviously remains the same.}.

	For $p^L_2 = \frac{p_0 + 2 p_1 + p_2}{4}$
	\begin{alignat}{2}
	       && \left\langle {p^L_2 - s, t} \right\rangle &\leq 0\nonumber\\
    \LRA && \left\langle {\frac{p_0 + 2p_1 + p_2}{4} - \frac{p_0 + 3p_1 + 3p_2 + p_3}{8}, \frac{3(-p_0 - p_1 + p_2 + p_3)}{4}} \right\rangle &\leq 0\nonumber\\
		\LRA && \langle {2p_0 + 4p_1 + 2p_2 - (p_0 + 3p_1 + 3p_2 + p_3), -p_0 - p_1 + p_2 + p_3} \rangle                              &\leq 0\nonumber\\
		\LRA && \underbrace{\langle {p_0 + p_1 - p_2 - p_3, -p_0 - p_1 + p_2 + p_3} \rangle}_{\leq 0}                                 &\leq 0.\nonumber
	\end{alignat}

	For $p^L_1 = \frac{p_0 + p_1}{2}$
	\begin{alignat}{2}
	       && \left\langle {p^L_1 - s, t} \right\rangle &\leq 0\nonumber\\
    \LRA && \left\langle {\frac{p_0 + p_1}{2} - \frac{p_0 + 3p_1 + 3p_2 + p_3}{8}, \frac{3(-p_0 - p_1 + p_2 + p_3)}{4}} \right\rangle                                                         &\leq 0\nonumber\\
		\LRA && \langle {4p_0 + 4p_1 - (p_0 + 3p_1 + 3p_2 + p_3), -p_0 - p_1 + p_2 + p_3} \rangle                                                                                      &\leq 0\nonumber\\
		\LRA && \langle {3p_0 + p_1 - 3p_2 - p_3, -p_0 - p_1 + p_2 + p_3} \rangle                                                                                                      &\leq 0\nonumber\\
		\LRA && \langle {3(p_2 - p_0) + (p_3 - p_1), (p_2 - p_0) + (p_3 - p_1)} \rangle                                                                                                &\geq 0\nonumber\\
		\LRA && 3\underbrace{\langle {p_2 - p_0, p_2 - p_0} \rangle}_{\geq 0} + 4\underbrace{\langle {p_2 - p_0, p_3 - p_1} \rangle}_{\geq 0\textrm{ by Eq.~(\ref{eq:cubic_bezier_p2p0p3p1})}} + \underbrace{\langle {p_3 - p_1, p_3 - p_1} \rangle}_{\geq 0} &\geq 0.\nonumber
	\end{alignat}

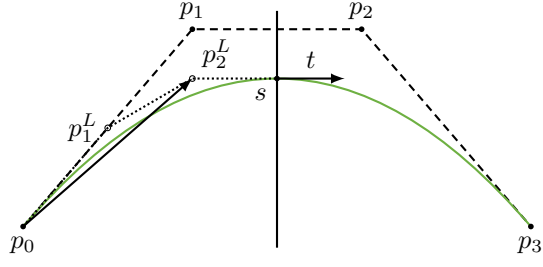
\begin{figure}[htbp]
\center
\begin{tikzpicture}[scale = 0.75]
\coordinate (A) at (0, 0);
\coordinate (B) at (3, 3.5);
\coordinate (C) at (6, 3.5);
\coordinate (D) at (9, 0);
\coordinate (AB) at ($0.5*(A)+0.5*(B)$);
\coordinate (BC) at ($0.5*(C)+0.5*(B)$);
\coordinate (CD) at ($0.5*(C)+0.5*(D)$);
\coordinate (ABBC) at ($0.5*(AB) + 0.5*(BC)$);
\draw[thick, densely dashed] (A) -- (B) -- (C) -- (D);
\draw[thick, HighlightColor] (A) .. controls (B) and (C)  .. (D);
\fill (A) circle (0.05);
\fill (B) circle (0.05);
\fill (C) circle (0.05);
\fill (D) circle (0.05);
\draw (AB) circle (0.05);
\draw (ABBC) circle (0.05);

			\coordinate (T) at ($-1.0*(A)-(B)+(C)+(D)$);
			\newdimen\xta
			\newdimen\yta
			\pgfextractx{\xta}{\pgfpointanchor{T}{center}}
			\pgfextracty{\yta}{\pgfpointanchor{T}{center}}
			\coordinate (N) at (-\yta, \xta);
			\coordinate (S) at ($1.0/8.0*(A)+3.0/8.0*(B)+3.0/8.0*(C)+1.0/8.0*(D)$);

\fill (S) circle (0.05);
\draw[thick] ($(S) - 0.25*(N)$) -- ($(S) + 0.1*(N)$);
\draw[thick, -latex] (S) -- ($(S) + 0.1*(T)$) node[sloped, midway, above, font=\small] {$t$};
\draw[thick, -latex, font=\small] (A) -- (ABBC);%
\node[anchor=north, font=\small] at (A) {$p_0$};
\node[anchor=east, font=\small] at (AB) {$p^L_1$};
\node[anchor=south west, font=\small] at (ABBC) {$p^L_2$};
\node[anchor=south, font=\small] at (B) {$p_1$};
\node[anchor=south, font=\small] at (C) {$p_2$};
\node[anchor=north, font=\small] at (D) {$p_3$};
\node[anchor=north east, font=\small] at (S) {$s$};
\draw[thick, densely dotted] (A) -- (AB) -- (ABBC) -- (S);
\end{tikzpicture}
\caption{$p_0$ must be located on the left side of the separation plane together with $p^L_2$ if $\left\langle {p^L_2 - p_0, t} \right\rangle
\geq 0$.}
\label{fig:cubic_bezier_p0}
\end{figure}

	Finally, knowing that $p^L_2$ is already on the correct side of the splitting plane (see also \Cref{fig:cubic_bezier_p0}), $p^L_0$ is
	on the same side of the plane if
	\begin{alignat}{2}
         && \left\langle {p^L_2 - p^L_0, t} \right\rangle                                                                                                                                                                                    & \geq 0\nonumber\\
		\LRA && \left\langle {\frac{p_0 + 2 p_1 + p_2}{4} - p_0, \frac{3\left(-p_0 -p_1 + p_2 + p_3 \right)}{4}} \right\rangle                                                                                                                   & \geq 0\nonumber\\
		\LRA && \langle {-3p_0 + 2 p_1 + p_2, -p_0 -p_1 + p_2 + p_3} \rangle                                                                                                                                                          & \geq 0\nonumber\\
		\LRA && \langle {2(p_1 - p_0) + (p_2 - p_0), (p_2 - p_0) + (p_3 - p_1 )} \rangle                                                                                                                                              & \geq 0\nonumber\\
		\LRA && 2 \underbrace{\langle {p_1 - p_0, p_2 - p_0} \rangle}_{\geq 0\textrm{ by Eq.~(\ref{eq:cubic_bezier_p2p0p1p0})}} + 2 \underbrace{\langle {p_1 - p_0, p_3 - p_1} \rangle}_{\geq 0\textrm{ by Eq.~(\ref{eq:cubic_bezier_p3p1p1p0})}} &\nonumber\\
         && + \underbrace{\langle {p_2 - p_0, p_2 - p_0} \rangle}_{\geq 0}  + \underbrace{\langle {p_2 - p_0, p_3 - p_1} \rangle}_{\geq 0\textrm{ by Eq.~(\ref{eq:cubic_bezier_p2p0p3p1})}}                                             & \geq 0.\nonumber
	\end{alignat}

For the right subcurve
\begin{alignat}{2}
       && \left\langle {p^R_1 - s, t} \right\rangle                                                                                              & \geq 0\nonumber\\
	\LRA && \left\langle {\frac{p_1 + 2p_2 + p_3}{4} - \frac{p_0 + 3p_1 + 3p_2 + p_3}{8}, \frac{3\left(-p_0-p_1+p_2+p_3 \right)}{4}} \right\rangle & \geq 0\nonumber\\
	\LRA && \langle {2p_1 + 4p_2 + 2p_3 - p_0 - 3p_1 - 3p_2 - p_3, -p_0 - p_1 + p_2 + p_3} \rangle                                      & \geq 0\nonumber\\
	\LRA && \underbrace{\langle { -p_0 - p_1 + p_2 + p_3, -p_0 - p_1 + p_2 + p_3} \rangle}_{\geq 0}                                     & \geq 0,\nonumber
\end{alignat}
\begin{alignat}{2}
       && \left\langle {p^R_2 - s, t} \right\rangle                                                                                                                                                                                                           & \geq 0\nonumber\\
	\LRA && \left\langle {\frac{p_2 + p_3}{2} - \frac{p_0 + 3p_1 + 3p_2 + p_3}{8}, \frac{3\left(-p_0-p_1+p_2+p_3 \right)}{4}} \right\rangle                                                                                                                     & \geq 0\nonumber\\
	\LRA && \langle {4p_2 + 4p_3 - p_0 - 3p_1 - 3p_2 - p_3, -p_0 - p_1 + p_2 + p_3} \rangle                                                                                                                                                          & \geq 0\nonumber\\
	\LRA && \langle {-p_0 - 3p_1 + p_2 + 3p_3, -p_0 - p_1 + p_2 + p_3} \rangle                                                                                                                                                                       & \geq 0\nonumber\\
	\LRA && \langle {(p_2 - p_0) + 3(p_3 - p_1), (p_2 - p_0) + (p_3 - p_1)} \rangle                                                                                                                                                                  & \geq 0\nonumber\\
	\LRA && \underbrace{\langle {p_2 - p_0, p_2 - p_0} \rangle}_{\geq 0} + 4\underbrace{\langle {p_2 - p_0, p_3 - p_1} \rangle}_{\geq 0\textrm{ by Eq.~(\ref{eq:cubic_bezier_p2p0p3p1})}} + 3 \underbrace{\langle {p_3 - p_1, p_3 - p_1} \rangle}_{\geq 0} & \geq 0,\nonumber
\end{alignat}
and again knowing that $p^R_1$ is already on the correct side of the splitting plane, $p^R_3$ is on the same side of the
plane if
\begin{alignat}{2}
       && \left\langle {p^R_3 - p^R_1, t} \right\rangle                                                                                                                                                                                    & \geq 0\nonumber\\
	\LRA && \left\langle {p_3 - \frac{p_1 + 2p_2 + p_3}{4}, \frac{3\left(-p_0 -p_1 + p_2 + p_3 \right)}{4}} \right\rangle                                                                                                                    & \geq 0\nonumber\\
	\LRA && \langle {-p_1 - 2p_2 + 3p_3, -p_0 -p_1 + p_2 + p_3} \rangle                                                                                                                                                           & \geq 0\nonumber\\
	\LRA && \langle {(p_3 - p_1) + 2 (p_3 - p_2), (p_2 - p_0) + (p_3 - p_1)} \rangle                                                                                                                                              & \geq 0\nonumber\\
	\LRA && \underbrace{\langle {p_3 - p_1, p_2 - p_0} \rangle}_{\geq 0\textrm{ by Eq.~(\ref{eq:cubic_bezier_p2p0p3p1})}} + \underbrace{\langle {p_3 - p_1, p_3 - p_1} \rangle}_{\geq 0}                                                &\nonumber\\
       && + 2\underbrace{\langle {p_3 - p_2, p_2 - p_0} \rangle}_{\geq 0\textrm{ by Eq.~(\ref{eq:cubic_bezier_p2p0p3p2})}} + 2\underbrace{\langle {p_3 - p_2, p_3 - p_1} \rangle}_{\geq 0\textrm{ by Eq.~(\ref{eq:cubic_bezier_p3p1p3p2})}} & \geq 0.\nonumber\\
\end{alignat}
	Finally, as a B\'ezier curve is a convex combination of its control points, the curve is always enclosed in their
convex hull, and therefore all subcurves in the two domains are completely on one side of the split plane each.
\end{proof}

\begin{lemma}
Splitting the cubic B\'ezier curve fulfilling the conditions (\ref{eq:cubic_bezier_p2p0p1p0}),
(\ref{eq:cubic_bezier_p3p1p1p0}), (\ref{eq:cubic_bezier_p3p1p3p2}), (\ref{eq:cubic_bezier_p2p0p3p2}),
(\ref{eq:cubic_bezier_p2p0p3p1}) in the middle results in two subcurves that also satisfy them all. Hence, all nested
subcurves from further bisection fulfill the conditions.
\end{lemma}

\begin{proof}
For the left subcurve
\begin{alignat}{2}
       && \left\langle {p_2^L - p_0^L, p_1^L - p_0^L} \right\rangle                                                                                &\geq 0\nonumber\\
	\LRA && \left\langle {\frac{p_0 + 2 p_1 + p_2}{4} - p_0, \frac{p_0 + p_1}{2} - p_0} \right\rangle                                                &\geq 0\nonumber\\
	\LRA && \langle {\frac{-3p_0 + 2 p_1 + p_2}{4}, \frac{p_1 - p_0}{2}} \rangle                                                          &\geq 0\nonumber\\
	\LRA && \langle {-3p_0 + 2 p_1 + p_2, p_1 - p_0} \rangle                                                                              &\geq 0\nonumber\\
	\LRA && \langle  {(p_2 - p_0) + 2(p_1 - p_0), p_1 - p_0} \rangle                                                                      &\geq 0\nonumber\\
	\LRA && \underbrace{\langle  {p_2 - p_0, p_1 - p_0} \rangle}_{\geq 0\textrm{ by Eq.~(\ref{eq:cubic_bezier_p2p0p1p0})}} + 2\underbrace{\langle {p_1 - p_0, p_1 - p_0} \rangle}_{\geq 0} &\geq 0\nonumber,
\end{alignat}
\begin{alignat}{2}
       && \left\langle {p_3^L - p_1^L, p_1^L - p_0^L} \right\rangle                                                                                                                                                                       &\geq 0\nonumber\\
	\LRA && \left\langle {\frac{p_0 + 3p_1 + 3p_2 + p_3}{8} - \frac{p_0 + p_1}{2}, \frac{p_0 + p_1}{2} - p_0} \right\rangle                                                                                                                 &\geq 0\nonumber\\
	\LRA && \langle {\frac{-3p_0 - p_1 + 3p_2 + p_3}{8}, \frac{p_1 - p_0}{2}} \rangle                                                                                                                                            &\geq 0\nonumber\\
	\LRA && \langle {-3(p_2 - p_0) + (p_3 - p_1), p_1 - p_0} \rangle                                                                                                                                                             &\geq 0\nonumber\\
	\LRA && -3 \underbrace{\langle {p_2 - p_0, p_1 - p_0} \rangle}_{\geq 0\textrm{ by Eq.~(\ref{eq:cubic_bezier_p2p0p1p0})}} + \underbrace{\langle {p_3 - p_1, p_1 - p_0} \rangle}_{\geq 0\textrm{ by Eq.~(\ref{eq:cubic_bezier_p3p1p1p0})}} &\geq 0\nonumber,
\end{alignat}
\begin{alignat}{2}
       && \left\langle {p_3^L - p_1^L, p_3^L - p_2^L} \right\rangle                                                                                                                                                                                          &\geq 0\nonumber\\
	\LRA && \left\langle {\frac{p_0 + 3p_1 + 3p_2 + p_3}{8} - \frac{p_0 + p_1}{2}, \frac{p_0 + 3p_1 + 3p_2 + p_3}{8} - \frac{p_0 + 2p_1 + p2}{4}} \right\rangle                                                                                                &\geq 0\nonumber\\
	\LRA && \langle {-3p_0 - p_1 + 3p_2 + p_3, -p_0 - p_1 + p_2 + p_3} \rangle                                                                                                                                                                      &\geq 0\nonumber\\
	\LRA && \langle {3(p_2 - p_0) + (p_3 - p_1), (p_2 - p_0) + (p_3 - p_1)} \rangle                                                                                                                                                                 &\geq 0\nonumber\\
	\LRA && 3\underbrace{\langle {p_2 - p_0, p_2 - p_0} \rangle}_{\geq 0} + 4\underbrace{\langle {p_2 - p_0, p_3 - p_1} \rangle}_{\geq 0\textrm{ by Eq.~(\ref{eq:cubic_bezier_p2p0p3p1})}} + \underbrace{\langle {p_3 - p_1, p_3 - p_1} \rangle}_{\geq 0} &\geq 0,\nonumber
\end{alignat}
\begin{alignat}{2}
       && \left\langle {p_2^L - p_0^L, p_3^L - p_2^L} \right\rangle                                                                                                                                                                        &\geq 0\nonumber\\
	\LRA && \left\langle {\frac{p_0 + 2 p_1 + p_2}{4} - p_0, \frac{p_0 + 3p_1 + 3p_2 + p_3}{8} - \frac{p_0 + 2p_1 + p2}{4}} \right\rangle                                                                                                    &\geq 0\nonumber\\
	\LRA && \langle {-3p_0 + 2 p_1 + p_2, -p_0 - p_1 + p_2 + p_3} \rangle                                                                                                                                                         &\geq 0\nonumber\\
	\LRA && \langle {(p_2 - p_0) + 2(p_1 - p_0), (p_2 - p_0) + (p_3 - p_1)} \rangle                                                                                                                                               &\geq 0\nonumber\\
	\LRA && \underbrace{\langle {p_2 - p_0, p_2 - p_0} \rangle}_{\geq 0} + \underbrace{\langle {p_2 - p_0, p_3 - p_1} \rangle}_{\geq 0\textrm{ by Eq.~(\ref{eq:cubic_bezier_p2p0p3p1})}}                                                &\nonumber\\
       && + 2\underbrace{\langle {p_1 - p_0, p_2 - p_0} \rangle}_{\geq 0\textrm{ by Eq.~(\ref{eq:cubic_bezier_p2p0p1p0})}} + 2\underbrace{\langle {p_1 - p_0, p_3 - p_1} \rangle}_{\geq 0\textrm{ by Eq.~(\ref{eq:cubic_bezier_p3p1p1p0})}} &\geq 0,\nonumber
\end{alignat}
\begin{alignat}{2}
       && \left\langle {p^L_2 - p^L_0, p^L_3 - p^L_1} \right\rangle                                                                                                                                                                                & \geq 0\nonumber\\
	\LRA && \left\langle {\frac{p_0 + 2 p_1 + p_2}{4} - p_0, \frac{p_0 + 3p_1 + 3p_2 + p_3}{8} - \frac{p_0 + p_1}{2}} \right\rangle                                                                                                                  & \geq 0\nonumber\\
	\LRA && \langle {-3p_0 + 2p_1 + p_2, p_0 + 3p_1 + 3p_2 + p_3 - 4p_0 - 4p_1} \rangle                                                                                                                                                   & \geq 0\nonumber\\
	\LRA && \langle {-3p_0 + 2p_1 + p_2, -3p_0 - p_1 + 3p_2 + p_3} \rangle                                                                                                                                                                & \geq 0\nonumber\\
	\LRA && \langle {2(p_1 - p_0) + (p_2 - p_0), 3(p_2 - p_0) + (p_3 - p_1)} \rangle                                                                                                                                                      & \geq 0\nonumber\\
	\LRA && 6\underbrace{\langle {(p_1 - p_0), (p_2 - p_0) } \rangle}_{\geq 0\textrm{ by Eq.~(\ref{eq:cubic_bezier_p2p0p1p0})}} + 2\underbrace{\langle {(p_1 - p_0), (p_3 - p_1) } \rangle}_{\geq 0\textrm{ by Eq.~(\ref{eq:cubic_bezier_p3p1p1p0})}} &\nonumber\\
       && +3\underbrace{\langle {(p_2 - p_0), (p_2 - p_0) } \rangle}_{\geq 0} + \underbrace{\langle {(p_2 - p_0), (p_3 - p_1) } \rangle}_{\geq 0\textrm{ by Eq.~(\ref{eq:cubic_bezier_p2p0p3p1})}}                                            & \geq 0.\nonumber
\end{alignat}

For the right subcurve 
\begin{alignat}{2}
       &&\left\langle  {p^R_2 - p^R_0, p^R_1 - p^R_0} \right\rangle                                                                                                                                                                                            & \geq 0\nonumber\\
	\LRA && \left\langle  {\frac{p_2 + p_3}{2} - \frac{p_0 + 3 p_1 + 3 p_2 + p_3}{8}, \frac{p_1 + 2p_2 + p_3}{4} - \frac{p_0 + 3p_1 + 3p_2 + p_3}{8}} \right\rangle                                                                                              & \geq 0\nonumber\\
	\LRA && \langle  {-p_0 - 3p_1 + p_2 + 3p_3, -p_0 - p_1 + p_2 + p_3} \rangle                                                                                                                                                                       & \geq 0\nonumber\\
	\LRA && \langle  {3(p_3 - p_1)+(p_2 - p_0), (p_3 - p_1)+(p_2 - p_0)} \rangle                                                                                                                                                                      & \geq 0\nonumber\\
	\LRA && 3\underbrace{\langle  {p_3 - p_1, p_3 - p_1} \rangle},{\geq 0} + 4\underbrace{\langle {p_3 - p_1, p_2 - p_0} \rangle}_{\geq 0\textrm{ by Eq.~(\ref{eq:cubic_bezier_p2p0p3p1})}} + \underbrace{\langle  {p_2 - p_0, p_2 - p_0} \rangle}_{\geq 0} & \geq 0,\nonumber
\end{alignat}
\begin{alignat}{2}
       &&\left\langle  {p^R_3 - p^R_1, p^R_1 - p^R_0} \right\rangle                                                                                                                                                                        & \geq 0\nonumber\\
	\LRA && \left\langle  {p_3 - \frac{p_1 + 2p_2 + p_3}{4}, \frac{p_1 + 2p_2 + p_3}{4} - \frac{p_0 + 3p_1 + 3p_2 + p_3}{8}} \right\rangle                                                                                                   & \geq 0\nonumber\\
	\LRA && \langle {-p_1 - 2p_2 + 3p_3, -p_0 -p_1 + p_2 + p_3} \rangle                                                                                                                                                           & \geq 0\nonumber\\
	\LRA && \langle {(p_3 - p_1) + 2(p-3 - p_2), (p_3 - p_1) + (p_2 - p_0)} \rangle                                                                                                                                               & \geq 0\nonumber\\
	\LRA && \underbrace{\langle {p_3 - p_1, p_3 - p_1} \rangle}_{\geq 0} + \underbrace{\langle {p_3 - p_1, p_2 - p_0} \rangle}_{\geq 0\textrm{ by Eq.~(\ref{eq:cubic_bezier_p2p0p3p1})}}                                                &\nonumber\\
       && + 2\underbrace{\langle {p_3 - p_2, p_3 - p_1} \rangle}_{\geq 0\textrm{ by Eq.~(\ref{eq:cubic_bezier_p3p1p3p2})}} + 2\underbrace{\langle {p_3 - p_2, p_2 - p_0} \rangle}_{\geq 0\textrm{ by Eq.~(\ref{eq:cubic_bezier_p2p0p3p2})}} & \geq 0,\nonumber
\end{alignat}
\begin{alignat}{2}
       &&\left\langle  {p^R_3 - p^R_1, p^R_3 - p^R_2} \right\rangle                                                                                                                        & \geq 0\nonumber\\
	\LRA && \left\langle {p_3 - \frac{p_1 - 2p_2 + p_3}{4}, p_3 - \frac{p_2 + p_3}{2}} \right\rangle                                                                                         & \geq 0\nonumber\\
	\LRA && \langle {-p_1 - 2p_2 + 3p_3, -p_2 + p_3} \rangle                                                                                                                      & \geq 0\nonumber\\
	\LRA && \langle {2(p_3 - p_2) + (p_3 - p_1), p_3 - p_2} \rangle                                                                                                               & \geq 0\nonumber\\
	\LRA && 2\underbrace{\langle {p_3 - p_2, p_3 - p_2} \rangle}_{\geq 0} + \underbrace{\langle {p_3 - p_1, p_3 - p_2} \rangle}_{\geq 0\textrm{ by Eq.~(\ref{eq:cubic_bezier_p3p1p3p2})}} & \geq 0,\nonumber
\end{alignat}
\begin{alignat}{2}
       &&\left\langle  {p^R_2 - p^R_0, p^R_3 - p^R_2} \right\rangle                                                                                                                                                                     & \geq 0\nonumber\\
	\LRA && \left\langle { \frac{p_2 + p_3}{2} - \frac{p_0 + 3 p_1 + 3 p_2 + p3}{8}, p_3 - \frac{p_2 + p_3}{2}} \right\rangle                                                                                                             & \geq 0\nonumber\\
	\LRA && \langle {-p_0 - 3p_1 + p_2 + 3 p_3, p_3 - p_2} \rangle                                                                                                                                                             & \geq 0\nonumber\\
	\LRA && \langle {3(p_3 - p_1) + (p_2 - p_0), p_3 - p_2} \rangle                                                                                                                                                            & \geq 0\nonumber\\
	\LRA && 3\underbrace{\langle {p_3 - p_1, p_3 - p_2} \rangle}_{\geq 0\textrm{ by Eq.~(\ref{eq:cubic_bezier_p3p1p3p2})}} + \underbrace{\langle {p_2 - p_0, p_3 - p_2} \rangle}_{\geq 0\textrm{ by Eq.~(\ref{eq:cubic_bezier_p2p0p3p2})}} & \geq 0,\nonumber
\end{alignat}
\begin{alignat}{2}
       &&\left\langle  {p^R_2 - p^R_0, p^R_3 - p^R_1} \right\rangle                                                                                                                                                                       & \geq 0\nonumber\\
	\LRA && \left\langle { \frac{p_2 + p_3}{2} - \frac{p_0 + 3 p_1 + 3 p_2 + p3}{8}, p_3 - \frac{p_1 + 2p_2 + p_3}{4}} \right\rangle                                                                                                        & \geq 0\nonumber\\
	\LRA && \langle {-p_0 - 3p_1 + p_2 + 3 p_3, 3p_3 - p_1 - 2 p_2} \rangle                                                                                                                                                      & \geq 0\nonumber\\
	\LRA && \langle {3(p_3 - p_1)+(p_2 - p_0), 2(p_3 - p_2) + (p_3 - p_1)} \rangle                                                                                                                                               & \geq 0\nonumber\\
	\LRA && 6\underbrace{\langle {p_3 - p_1, p_3 - p_2} \rangle}_{\geq 0\textrm{ by Eq.~(\ref{eq:cubic_bezier_p3p1p3p2})}} + 3\underbrace{\langle {p_3 - p_1, p_3 - p_1} \rangle}_{\geq 0}                                             &\nonumber\\
       && + 2\underbrace{\langle {p_2 - p_0, p_3 - p_2} \rangle}_{\geq 0\textrm{ by Eq.~(\ref{eq:cubic_bezier_p2p0p3p2})}} + \underbrace{\langle {p_2 - p_0, p_3 - p_1} \rangle}_{\geq 0\textrm{ by Eq.~ (\ref{eq:cubic_bezier_p2p0p3p1})}} & \geq 0.\nonumber
\end{alignat}
\end{proof}

\section{Intersection code}

\begin{lstlisting}[caption={Intersection of a unit ray with a plane through $p$ with normal $n$.},label=lst:ray-plane]
float intersect_plane(const vec3f p, const vec3f n)
{
	return dot3(n, p) / n.z;
}
\end{lstlisting}

\if 0
\begin{lstlisting}[caption=Intersection of a unit ray with an infinite cylinder with radius $r$ through $o$ along $a$.,label=lst:ray-cylinder]
tuple<float, float> intersect_cylinder
(
	const vec3f o,
	const vec3f a,
	const float r
)
{
	const float d = a.x * o.y - a.y * o.x;
	const float g = a.x * a.x + a.y * a.y;
	const float e = r * r - d * d / g;

	if (g == 0 && o.x * o.x + o.y * o.y < r * r)
		return make_tuple(-FLT_MAX, FLT_MAX);
	if (e < 0)
		return make_tuple(FLT_MAX, FLT_MAX);

	const float t_cpa = o.z - a.z * (a.x*o.x + a.y*o.y)/g;
	const float s = sqrtf(e * (a.z * a.z + g) / g);

	return make_tuple(t_cpa - s, t_cpa + s);
}
\end{lstlisting}
\else
\begin{minipage}{\linewidth}
\begin{lstlisting}[caption=Intersection of a unit ray with an infinite cylinder with radius $r$ through $o$ along $a$.,label=lst:ray-cylinder]
tuple<float, float> intersect_cylinder
(
	const vec3f o,
	const vec3f a,
	const float r
)
{
	const float d = a.x * o.y - a.y * o.x;
	const float g = a.x * a.x + a.y * a.y;
	if (g == 0 && o.x * o.x + o.y * o.y < r * r)
		return make_tuple(-FLT_MAX, FLT_MAX);

	const float h = 1.0f / g;
	const float e = r * r - d * d * h;

	if (e < 0)
		return make_tuple(FLT_MAX, FLT_MAX);

	const float t_cpa = o.z - a.z * (a.x*o.x + a.y*o.y)*h;
	const float s = sqrtf(e * (a.z * a.z + g) * h);

	return make_tuple(t_cpa - s, t_cpa + s);
}
\end{lstlisting}
\end{minipage}
\fi

\begin{minipage}{\linewidth}
\begin{lstlisting}[caption={Maximum distance of the two points $p, q$ to the line starting in the origin and going along $d$.},label=lst:distance-point-line]
float dist_2points_line
(
	const vec3f& p,
	const vec3f& q,
	const vec3f& d
)
{
	const float c = dot3(d, d);

	const float bp = dot3(p, d) / c;
	const vec3f pPbp = p - bp * d;
	const float l2_pPbp = dot3(pPbp, pPbp);

	const float bq = dot3(q, d) / c;
	const vec3f qPbq = q - bq * d;
	const float l2_qPbq = dot3(qPbq, qPbq);

	return sqrtf(max(l2_pPbp, l2_qPbq));
}
\end{lstlisting}
\begin{lstlisting}[caption={Conversion of the integer representation of the interval (start, size) to the floating point
interval $\lbrack u_0, u_1\rbrack$.}, label={lst:calculate_interval}]
tuple<float, float> get_interval
(
 const int32_t start,
 const int32_t size
)
{
	const int32_t ui0  = 0x3f800000 | start;
	const float u0 = int_as_float(ui0) - 1.0f;
	const int32_t ui1 = min(ui0 + size, 0x40000000);
	const float u1 = int_as_float(ui1) - 1.0f;
	return make_tuple(u0, u1);
}
\end{lstlisting}
\begin{lstlisting}[caption={Evaluate cubic B\'ezier curve and its (scaled) derivative.},label=lst:eval_cubic_bezier]
vec4f eval(const BezierCurve c, const float u)
{
	float v = 1.0f - u;
	return v * v * v * c.p0
	       + 3.0f * u * v * v * c.p1
	       + 3.0f * u * u * v * c.p2
	       + u * u * u * c.p3;
}

vec4f eval_derivative(const BezierCurve c, const float u)
{
	float v = 1.0f - u;
	return v * v * (c.p1 - c.p0)
	       + 2.0f * u * v * (c.p2 - c.p1)
	       + u * u * (c.p3 - c.p2);
}
\end{lstlisting}
\end{minipage}

\begin{minipage}{\linewidth}
\begin{lstlisting}[caption={Re-calculation of the cubic B\'ezier curve
represented by the tuple $(p, d, t_0,t_1)$ for an interval $\lbrack u0, u1 \rbrack$ used after backtracking.},
label=lst:recalculation, morekeywords={BezierCurve, BezierCurveDelta, Intersection}]
BezierCurveDelta calculate_control_points
(
	const BezierCurve& c,
	const uint32_t     cur_start,
	const uint32_t     cur_size
)
{
	float u0, u1;
	tie(u0, u1) = get_interval(cur_start, cur_size);
	const vec4f p = eval(c, u0);
	const vec4f d = eval(c, u1) - p;
	const vec4f t0 = (u1 - u0) * eval_derivative(c, u0);
	const vec4f t1 = (u1 - u0) * eval_derivative(c, u1);
	return BezierCurveDelta(p, d, t0, t1);
}
\end{lstlisting}
\begin{lstlisting}[caption={Subdivide the curve}, label={lst:subdivide}, morekeywords={BezierCurve, BezierCurveDelta, Intersection}]
void subdivide_and_set
(
	const bool        go_right,
	const vec4f&      center,
	const vec4f&      t_center,
	BezierCurveDelta& c
)
{
	if (go_right)
	{
		c.p  += center;
		c.d  -= center;
		c.t0  = t_center;
		c.t1 *= 0.5f;
	}
	else
	{
		c.d   = center;
		c.t0 *= 0.5f;
		c.t1  = t_center;
	}
}
\end{lstlisting}
\begin{lstlisting}[caption={Calculate a conservative radius of a cylinder bounding a cubic B\'ezier curve.}, label={lst:calc_radius}, morekeywords={BezierCurve, BezierCurveDelta, Intersection}]
float calculate_conservative_radius
(
	const BezierCurveDelta& c
)
{
	const float dist = dist_2points_line(c.t0.xyz(), c.t1.xyz(), c.d.xyz());
	const float max_r = c.p.w + max(max(0.0f, c.t0.w), max(c.d.w, c.d.w - c.t1.w));
	return dist + max_r;
}
\end{lstlisting}
\end{minipage}

\begin{minipage}{\linewidth}
\begin{lstlisting}[caption={Partition space, subdivide curve, and update state}, label={lst:subdivide_partition_and_update}, morekeywords={BezierCurve, BezierCurveDelta, Intersection}]
tuple<bool, bool> subdivide_partition_and_update
(
	const BezierCurveDelta& c,
	const float             t0,
	const float             t1,
	float&                  t_min,
	float&                  t_max
)
{
	vec4f center;
	vec4f t_center;
	tie(center, t_center) = c.get_center_and_tangent();
	const float t_plane = intersect_plane(c.p.xyz() + center.xyz(), t_center.xyz());

	const bool go_right = t_plane > t0 ^ t_center.z > 0;
	const bool both_hit = t0 < t_plane && t1 > t_plane;

	// Update t interval
	if (t_plane > t0) t_max = min(t_max, t_plane);
	else t_min = max(t_min, t_plane);

	// Subdivide
	subdivide_and_set(go_right, center, t_center, c);

	return make_tuple(both_hit, go_right);
}
\end{lstlisting}
\begin{lstlisting}[caption={Calculate the valid $t$ interval between the cropping planes.}, label={lst:calc_t_interval}, morekeywords={BezierCurve, BezierCurveDelta, Intersection}]
tuple<float, float> calculate_t_interval
(
	const BezierCurveDelta& cur_c,
	const float             ray_t_max
)
{
	const float t0 = intersect_plane(cur_c.p, cur_c.t0);
	const float t1 = intersect_plane(cur_c.p + cur_c.d, cur_c.t1);
	float t_min = 0.0f;
	if (cur_c.t0.z > 0) t_min = max(t_min, t0);
	if (cur_c.t1.z < 0) t_min = max(t_min, t1);
	float t_max = ray_t_max;
	if (cur_c.t0.z < 0) t_max = min(t_max, t0);
	if (cur_c.t1.z > 0) t_max = min(t_max, t1);

	return make_tuple(t_min, t_max);
}
\end{lstlisting}
\end{minipage}

\begin{minipage}{\linewidth}
\begin{lstlisting}[caption={Transform a cubic B\'ezier curve to a ray frame.}, label={lst:transform_curve}, morekeywords={BezierCurve, BezierCurveDelta, Intersection}]
vec3f transform
(
	const vec3f& p,
	const vec3f& origin,
	const vec3f& u,
	const vec3f& v,
	const vec3f& w
)
{
	vec3f q = p - origin;
	return vec3f(dot3(q, u), dot3(q, v), dot3(q, w));
}

BezierCurve transform_to_ray_frame
(
	const vec3f&       origin,
	const vec3f&       direction,
	const BezierCurve& c
)
{
	const vec3f w = direction;
	vec3f u, v;
	tie(u, v) = make_ONB(w);

	return BezierCurve(transform(c.p0, origin, u, v, w),
	                   transform(c.p1, origin, u, v, w),
	                   transform(c.p2, origin, u, v, w),
	                   transform(c.p3, origin, u, v, w));
}
\end{lstlisting}
\begin{lstlisting}[caption={Bit string manipulation to advance to the next finer level (go\_down) and for backtracking (jump\_up).}, label={lst:bitstring_manipulation}]
void go_down
(
	const bool both_hit,
	const bool go_right,
	uint32_t&  cur_size,
	uint32_t&  bit_string,
	uint32_t&  cur_start
)
{
	cur_size /= 2;
	if (both_hit) bit_string ^= cur_size;
	if (go_right) cur_start  ^= cur_size;
}

void jump_up
(
	uint32_t& cur_size,
	uint32_t& cur_start,
	uint32_t& bit_string
)
{
	// ctz = count trailing zeros
	cur_size    = 1 << ctz(bit_string);
	cur_start  ^= cur_size;
	bit_string ^= cur_size;
	cur_start  &= ~(cur_size - 1);
}
\end{lstlisting}
\end{minipage}

\begin{minipage}{\linewidth}
\begin{lstlisting}[caption={Calculate the normal and $u$ parameter in the intersection.}, label={lst:calc_intersection}, morekeywords={BezierCurve, BezierCurveDelta, Intersection}]
Intersection calculate_intersection
(
	const uint32_t          cur_start,
	const uint32_t          cur_size,
	const float             t,
	const BezierCurve&      c,
	const BezierCurveDelta& cur_c
)
{
		// Project intersection - p_0 onto p_n - p_0
		vec3f hb = cur_c.p.xyz() - vec3f(0, 0, t);
		float u_local = dot3(hb, cur_c.d.xyz());
		u_local /= dot3(cur_c.d.xyz(), cur_c.d.xyz());
		u_local = max(0.0f, min(1.0f, -u_local));

		float u0, u1;
		tie(u0, u1) = get_interval(cur_start, cur_size);
		float u = u0 + u_local * (u1 - u0);

		// End caps?
		if (u == 0.0f || u == 1.0f)
		{
			vec3f n = (u == 0.0f ? c.p0.xyz() : c.p3.xyz()) -
			          (u == 0.0f ? c.p1.xyz() : c.p2.xyz());
			return Intersection(t, u, normalize3(n));
		}

		vec3f ap = cur_c.p.xyz() + u_local * cur_c.axis;

		// Recompute frame
		vec3f frame_u, frame_v;
		tie(frame_u, frame_v) = make_ONB(ray_direction);

		// Caclulate normal in global frame
		vec3f normal = normalize3(t * ray_direction - ap.x * frame_u - ap.y * frame_v - ap.z * ray_direction);

		return Intersection(t, u, normal);
}
\end{lstlisting}
\end{minipage}

\begin{minipage}{\linewidth}
\begin{lstlisting}[caption={Iterative routine for ray/fiber intersection using recursive subdivision with disjoint
bounding volumes.}, label={lst:algorithm}, morekeywords={BezierCurve, BezierCurveDelta, Intersection, Ray}]
Intersection intersect
(
	const Ray&         ray,
	const BezierCurve& c,
	const uint32_t     min_size,
)
{
	// State
	const BezierCurve c_local = transform_to_ray_frame(ray.origin, ray.direction, c);
	BezierCurveDelta cur_c    = c_local; // conversion {p0, p1, p2, p3} -> {p, d, t0, t1}
	float t                   = FLT_MAX;
	uint32_t bit_string       = 0;
	uint32_t cur_size         = 1 << 23;
	uint32_t cur_start        = 0;
	float t_min, t_max;

	// Set the initial t interval
	tie(t_min, t_max) = calculate_t_interval(cur_c, ray.t_max);

	while (true)
	{
		const float cur_r = calculate_conservative_radius(cur_c);

		float t0, t1;
		tie(t0, t1) = intersect_cylinder(cur_c.p, cur_c.d, cur_r);
		if (t1 >= t_min && t0 <= t_max)
		{
			if (cur_size <= min_size)
			{
				t = t0;
				break;
			}

			// Subdivide and go down in the hierarchy
			bool both_hit, go_right;
			tie(both_hit, go_right) = subdivide_partition_and_update(cur_c, t0, t1, t_min, t_max);
			go_down(both_hit, go_right, cur_size, bit_string, cur_start);
		}
		else
		{
			// Done?
			if (!bit_string) break;

			// Backtracking
			jump_up(cur_size, cur_start, bit_string);
			cur_c = calculate_control_points(c_local, cur_start, cur_size);

			// Re-calculate the t interval
			tie(t_min, t_max) = calculate_t_interval(cur_c);
		}
	}

	// Calculate normal etc. if a closer intersection has been found
	if (t < ray.t_max)
		return calculate_intersection(cur_start, cur_size, t, c, cur_c);
	else
		return Intersection(); // no intersection
}
\end{lstlisting}
\end{minipage}


\begin{thebibliography}{WBW{\etalchar{+}}14}

\bibitem[BB82]{Ballard:1982}
D.~Ballard and C.~Brown.
\newblock {\em Computer Vision}.
\newblock Prentice Hall Professional Technical Reference, 1st edition, 1982.

\bibitem[BGA12]{Barringer:2012}
R.~Barringer, C.~Gribel, and T.~Akenine{-}M{\"{o}}ller.
\newblock High-quality curve rendering using line sampled visibility.
\newblock {\em {ACM} Trans. Graph.}, 31(6):162:1--162:10, 2012.

\bibitem[BK85]{Bronsvoort:1985}
W.~Bronsvoort and F.~Klok.
\newblock Ray tracing generalized cylinders.
\newblock {\em ACM Trans. Graph.}, 4(4):291--303, October 1985.

\bibitem[Boy68]{Boyer:1968}
C.~Boyer.
\newblock {\em A History of Mathematics}.
\newblock Wiley, New York, 1968.

\bibitem[Cat74]{Catmull:1974}
E.~Catmull.
\newblock {\em A Subdivision Algorithm for Computer Display of
  Curved\\Surfaces.}
\newblock PhD thesis, The University of Utah, 1974.
\newblock AAI7504786.

\bibitem[CW94]{Cychosz:1994}
J.~Cychosz and W.~Waggenspack, Jr.
\newblock Intersecting a ray with a cylinder.
\newblock In Paul~S. Heckbert, editor, {\em Graphics Gems IV}, pages 356--365.
  Academic Press Professional, Inc., San Diego, CA, USA, 1994.

\bibitem[DBC{\etalchar{+}}17]{duff:onb:2017}
T.~Duff, J.~Burgess, P.~Christensen, C.~Hery, A.~Kensler, M.~Liani, and
  R.~Villemin.
\newblock Building an orthonormal basis, revisited.
\newblock {\em Journal of Computer Graphics Techniques (JCGT)}, 6(1):1--8,
  2017.

\bibitem[dC59]{Casteljau:1959}
P.~de~Casteljau.
\newblock Outillages methodes calcul.
\newblock Technical report, Citro{\"e}n France, 1959.

\bibitem[Fri12]{frisvad:onb:2012}
J.~Frisvad.
\newblock Building an orthonormal basis from a 3d unit vector without
  normalization.
\newblock {\em Journal of Graphics Tools}, 16(3):151--159, 2012.

\bibitem[HARL05]{Hain:2005}
T.~Hain, A.~Ahmad, S.~Racherla, and D.~Langan.
\newblock Fast, precise flattening of cubic b{\'e}zier path and offset curves.
\newblock {\em Comput. Graph.}, 29(5):656--666, October 2005.

\bibitem[Lei95]{Leipelt:1995}
A.~Leipelt.
\newblock Ray tracing a swept sphere.
\newblock In Alan~W. Paeth, editor, {\em Graphics Gems V}, pages 258--267.
  Academic Press, 1995.

\bibitem[NO02]{Nakamaru:2002}
K.~Nakamaru and Y.~Ohno.
\newblock Ray tracing for curves primitive.
\newblock In {\em WSCG}, pages 311--316, 2002.

\bibitem[QCH{\etalchar{+}}14]{Qin:2014}
H.~Qin, M.~Chai, Q.~Hou, Z.~Ren, and K.~Zhou.
\newblock Cone tracing for furry object rendering.
\newblock {\em IEEE Transactions on Visualization and Computer Graphics},
  20(8):1178--1188, August 2014.

\bibitem[Res17]{Reshetov:2017}
A.~Reshetov.
\newblock Exploiting budan-fourier and vincent's theorems for ray tracing 3d
  b\'{e}zier curves.
\newblock In {\em Proceedings of High Performance Graphics}, HPG '17, pages
  5:1--5:11, New York, NY, USA, 2017. ACM.

\bibitem[SFD09]{Stich:2009}
M.~Stich, H.~Friedrich, and A.~Dietrich.
\newblock Spatial splits in bounding volume hierarchies.
\newblock In {\em Proceedings of the Conference on High Performance Graphics
  2009}, HPG '09, pages 7--13, New York, NY, USA, 2009. ACM.

\bibitem[Smi29]{Smith:1929}
D.~Smith.
\newblock {\em {A} {S}ource {B}ook in {M}athematics}.
\newblock McGraw-Hill Book Company, Inc., London, 1929.

\bibitem[vW84]{vanWijk:1984}
J.~van Wijk.
\newblock Ray tracing objects defined by sweeping planar cubic splines.
\newblock {\em ACM Trans. Graph.}, 3(3):223--237, July 1984.

\bibitem[vW85]{vanWijk:1985}
J.~van Wijk.
\newblock Ray tracing objects defined by sweeping a sphere.
\newblock {\em Computers {\&} Graphics}, 9(3):283--290, 1985.

\bibitem[WBW{\etalchar{+}}14]{Woop:2014}
S.~Woop, C.~Benthin, I.~Wald, G.~Johnson, and E.~Tabellion.
\newblock Exploiting local orientation similarity for efficient ray traversal
  of hair and fur.
\newblock In {\em Proceedings of High Performance Graphics}, HPG '14, pages
  41--49, Aire-la-Ville, Switzerland, Switzerland, 2014. Eurographics
  Association.

\end{thebibliography}
\end{document}